\documentclass[11pt]{article}
\usepackage[utf8]{inputenc}
\usepackage{a4wide}
\usepackage{graphicx}
\usepackage{amsmath}
\usepackage{amssymb}
\usepackage{amsthm}
\usepackage{tensor}
\usepackage{mathrsfs}
\usepackage{cleveref}
\usepackage{mathtools}
\usepackage{titling}
\usepackage{tocloft}

\author{Szymon Bagiński  and Jacek Jezierski \\ Department of Mathematical Methods in Physics \\ Faculty of Physics, University of Warsaw \\ Pasteura 5, 02-093, Warsaw, Poland}

\title{Comparison of definitions of angular momentum at null infinity}

\renewcommand{\d}{\mathrm{d}}
\newcommand{\scri}{\mathscr{I}}
\newcommand{\cjk}{J^{(\text{CJK})}}
\newcommand{\cn}{J^{(\text{CN})}}
\newcommand{\cwy}{J^{(\text{CWY})}}
\newcommand{\dcjk}{\dot{J}^{(\text{CJK})}}
\newcommand{\dcn}{\dot{J}^{(\text{CN})}}
\newcommand{\dcwy}{\dot{J}^{(\text{CWY})}}

\theoremstyle{definition}
\newtheorem{definition}{Definition}[section]
\newtheorem{theorem}{Theorem}[section]
\newtheorem{lemma}{Lemma}[section]

\allowdisplaybreaks

\begin{document}

\maketitle

\begin{abstract}
    The supertranslation ambiguity of angular momentum at null infinity creates a technical difficulty in choosing the proper definition for this conserved quantity. As a result, many different approaches are found in the literature. This article aims to review the most relevant of these notions and compare their most significant properties. A new definition of angular momentum at null infinity is also proposed, which is supertranslation invariant under certain assumptions.
\end{abstract}

\tableofcontents

\section{Introduction}\label{introduction}

After the successful detection of gravitational waves by LIGO \cite{abbott2016observation}, the topic of gravitational radiation became even more important in the theoretical understanding of general relativity. Currently, researchers usually describe radiating systems by analysing how different physical quantities behave asymptotically, which resembles how a given system is perceived by a distant observer. To make any valid conclusions, one needs to first define certain conserved quantities (e.g., energy, momentum, angular momentum) in a meaningful way, and then find expressions for their fluxes at the so-called null infinity. Although the matter is resolved in the case of energy and momentum \cite{cjk}, there are still some major technical difficulties regarding angular momentum.

The problem arises in how the radiation zone is modelled mathematically - it is usually taken to be the boundary, called the null infinity, of a conformal completion of a given spacetime. Although conformal completions are constructed in such a way that they should resemble Minkowski spacetime at large distances, the group of symmetries of the resulting null infinity is significantly larger than the Poincaré group. In fact, it is an infinite-dimensional group called the Bondi-Metzner-Sachs (BMS) group. The BMS group consists not only of translations and rotations but also contains the so-called supertranslations \cite{sachs}. The presence of such a kind of transformations makes defining conserved quantities, which, given Noether's theorem, have to be associated with certain symmetries of a physical system, fundamentally difficult. As a result, one encounters a supertranslation ambiguity when trying to define angular momentum at null infinity. Hence, many different conventions may now be found in the literature.

The goal of this thesis is to review three different possible definitions of angular momentum at null infinity and to compare their properties which may be of importance when dealing with physical systems. Firstly, essential definitions and theorems are reviewed. Then, the concerned definitions are presented and analysed in how they differ from each other. Next, their fluxes and transformation rules under supertranslations are calculated. Finally, we present a new definition of angular momentum which, under a suitable assumption, is supertranslation invariant.

\subsection{Conventions and notation}

By a spacetime we mean a 4-dimensional smooth manifold with a smooth Lorentzian metric. The signature $(-+++)$ is used for the spacetime metric. All fields and functions are assumed to be smooth unless stated otherwise. When working with tensor components, lower Greek indices ($\mu$, $\nu$, etc.) enumerate all 4 coordinates and upper Latin indices ($A$, $B$, etc.) enumerate the spherical coordinates, which are also abbreviated with $(\tensor{x}{^A})$. $\tensor{h}{_A_B}$ denotes the standard round metric on $S^2$, and all upper Latin indices are raised and lowered with respect to it. $\tensor{\epsilon}{_A_B}$ denotes the corresponding volume form. $D^A$ and $||A$ denote the Levi-Civita covariant derivative associated with $\tensor{h}{_A_B}$. Its Riemann tensor, Ricci tensor, and Ricci scalar are labelled respectively by $\tensor{R}{^A_B_C_D}$, $\tensor{R}{_A_B}$, and $R$. The Laplacian is denoted with $D^2 = D_AD^A$. On $S^2$, integration with respect to the standard volume form associated with $\tensor{h}{_A_B}$ is assumed implicitly. Partial derivatives with respect to the retarded time coordinate are denoted with a dot: $\tensor{\partial}{_u} f \equiv \dot{f}$.

\section{Background information}

Here, we present the concepts of Bondi-Sachs coordinates and conformal completion. Moreover, we recall several useful facts from differential geometry on $S^2$.

\subsection{Bondi-Sachs coordinates and null infinity}

\begin{definition}[Bondi-Sachs coordinates]
    Coordinates $(u, r, \tensor{x}{^A})$, where $u \in (u_{-}, u_{+})$ is the retarded time coordinate, $r \in (r_0, \infty)$ is the radial coordinate, and $(\tensor{x}{^A})$ are local coordinates on $S^2$, defined on some open subset of a spacetime $(M, g)$ are called Bondi-Sachs coordinates if
    \begin{equation}\label{eq:bs-metric}
        \tensor{g}{_\mu_\nu} \d\tensor{x}{^\mu} \d\tensor{x}{^\nu} = - Ve^{2\beta}\d{u^2} - 2e^{2\beta}\d{u}\d{r} + r^2\tensor{\Tilde{g}}{_A_B} \left( \d\tensor{x}{^A} - \tensor{U}{^A}\d{u} \right) \left( \d\tensor{x}{^B} - \tensor{U}{^B}\d{u} \right)
    \end{equation}
    \begin{equation*}
        \text{and}\quad\frac{\partial \left( \det \tensor{\Tilde{g}}{_A_B} \right)}{\partial r} = 0 ,
    \end{equation*}
    for some functions $V$, $\beta$, a $S^2$-vector field $\tensor{U}{^A}$, and a $S^2$-Riemannian metric $\tensor{\Tilde{g}}{_A_B}$.
\end{definition}

\begin{definition}[Conformal completion]
    Let $(M, g)$ be a spacetime such that there exists an open subset of $M$ on which Bondi-Sachs coordinates $(u, r, \tensor{x}{^A})$ can be introduced. A spacetime $\Tilde{M} = M \cup \scri^+$, where $\scri^+ \equiv \{ \frac{1}{r} = 0 \}$ (called the future null infinity), is called a conformal completion of $M$ if
    \begin{enumerate}
        \item $(u, r, \tensor{x}{^A})$, $V$, $\beta$, $\tensor{U}{^A}$, and $\tensor{\Tilde{g}}{_A_B}$ extend smoothly by continuity to $\scri^+$,
        \item $\displaystyle\lim_{r\to\infty} \tensor{\Tilde{g}}{_A_B} = \tensor{h}{_A_B}$, where $\tensor{h}{_A_B}$ is the standard round metric on $S^2$.
    \end{enumerate}
\end{definition}

In order to model some physical objects using these two definitions, we need to know more about the metric coefficients in (\ref{eq:bs-metric}). This is provided by the following theorem.

\begin{theorem}
    In a conformal completion, the coefficients of the metric from (\ref{eq:bs-metric}) have the following expansions:
    \begin{align*}
        \tensor{\Tilde{g}}{_A_B} &= \tensor{h}{_A_B} \left( 1 + \frac{1}{4r^2}\tensor{\chi}{^C^D}\tensor{\chi}{_C_D} \right) + \frac{\tensor{\chi}{_A_B}}{r} + O(r^{-3}) , \\
        \beta &= - \frac{\tensor{\chi}{_A_B}\tensor{\chi}{^A^B}}{32r^2} + O(r^{-3}) , \\
        \tensor{U}{^A} &= - \frac{\tensor{\chi}{^A^B_{||B}}}{2r^2} + \frac{2\tensor{N}{^A}}{r^3} + \frac{\left( \tensor{\chi}{^B^C}\tensor{\chi}{_B_C} \right)_{||A}}{16r^3} + \frac{\tensor{\chi}{^A^B}\tensor{\chi}{_B_C^{||C}}}{2r^3} + O(r^{-4}) , \\
        V &= 1 - \frac{2m}{r} + \frac{\tensor{\chi}{^A^B_{||B}}\tensor{\chi}{_B_C^{||C}} - 4\tensor{N}{^A_{||A}}}{4r^2} + \frac{\tensor{\chi}{^A^B}\tensor{\chi}{_A_B}}{16r^2} + O(r^{-3}) ,
    \end{align*}
    where $\tensor{\chi}{_A_B}$, $\tensor{N}{_A}$, and $m$ depend on $(u, \tensor{x}{^A})$. $m$ is called the mass aspect, $\tensor{N}{_A}$ is called the angular aspect, and $\tensor{\chi}{_A_B}$ is called the shear tensor. Moreover, $\tensor{\chi}{_A_B}$ is traceless and symmetric.
\end{theorem}
\begin{proof}
    See \cite[page 7]{jezierski1998} and \cite{van1966gravitational}.
\end{proof}

$\tensor{\chi}{_A_B}$, $\tensor{N}{_A}$, and $m$ define the asymptotical behaviour of the spacetime in question and hence are called the data on $\scri^+$. They will be later used in order to construct different notions of angular momentum. As a result, we are particularly interested in their properties. Most importantly, due to Einstein's equations, they must satisfy certain evolution equations.

\begin{theorem}
    The functions $m$ and $\tensor{N}{_A}$ satisfy the following equations:
    \begin{align}
        \dot{m} &= - \frac18 \tensor{\dot{\chi}}{_A_B}\tensor{\dot{\chi}}{^A^B} + \frac14 \tensor{\dot{\chi}}{^A^B_{||BA}}, \label{eq:m-flux} \\
        3\tensor{\dot{N}}{_A} &= - \tensor{m}{_{||A}} + \frac14 \tensor{\epsilon}{_A_B}\tensor{\Tilde{\lambda}}{^{||B}} - \frac34 \tensor{\chi}{^A_B}\tensor{\dot{\chi}}{^B^C_{||C}} - \frac14 \tensor{\dot{\chi}}{^C^D}\tensor{\chi}{^A_C_{||D}}, \label{eq:n-flux} \\
        \Tilde{\lambda} &= h^{BD}\epsilon^{AC}\chi_{DA||BC}. \nonumber
    \end{align}
\end{theorem}
\begin{proof}
    See \cite[page 92]{cjk}.
\end{proof}

\subsection{Supertranslations}

As mentioned in Chapter \ref{introduction}, the group of symmetries of $\scri^+$ is not the Poincaré group, as is the case in the Minkowski spacetime, but the BMS group, which extends the Poincaré group to include the so-called supertranslations.

\begin{definition}[Supertranslation]
    A supertranslation is a change of coordinates \\ $(u, \tensor{x}{^A}) \mapsto (u - f(\tensor{x}{^A}), \tensor{x}{^A})$, where $f: S^2 \rightarrow \mathbb{R}$.
\end{definition}

Also mentioned in Chapter \ref{introduction} was how this fact gives rise to the supertranslation ambiguity of angular momentum. Hence, when comparing different definitions of angular momentum, we would like to know how they are transformed under supertranslations. In order to calculate that, we will need the following theorem.

\begin{theorem}\label{thm:data-transformations}
    For any function $f: S^2 \rightarrow \mathbb{R}$, the data on $\scri^+$ obey the following transformation rules under the supertranslation $u \mapsto u - f(\tensor{x}{^A})$, $\tensor{x}{^A} \mapsto \tensor{x}{^A}$:
    \begin{align}
        \tensor{\chi}{_A_B} &\longmapsto \tensor{\chi}{_A_B} - 2\tensor{f}{_{||AB}} + \tensor{h}{_A_B}D^2f, \label{eq:chi-trans} \\
        m &\longmapsto m + \frac12 \tensor{\dot{\chi}}{^A^B_{||B}}\tensor{f}{_{||A}} + \frac14 \tensor{\dot{\chi}}{^A^B}\tensor{f}{_{||AB}} + \frac14 \tensor{\ddot{\chi}}{^A^B}\tensor{f}{_{||A}}\tensor{f}{_{||B}}, \nonumber \\
        \begin{split}
            \tensor{N}{_A} &\longmapsto \tensor{N}{_A} - m\tensor{f}{_{||A}} + \frac14 \tensor{f}{_{||C}}\tensor{\chi}{_A_B^{||BC}} - \frac14 \tensor{f}{_{||C}}\tensor{\chi}{^C^B_{||BA}} + \frac14 \tensor{\dot{\chi}}{_A_B^{||B}}\tensor{f}{_{||C}}\tensor{f}{^{||C}} \\
            &\phantom{=}\quad - \frac14 \tensor{\dot{\chi}}{^B^C}\tensor{\chi}{_A_B}f_{||C} - \frac16 \tensor{\ddot{\chi}}{^B^C}\tensor{f}{_{||B}}\tensor{f}{_{||C}}\tensor{f}{_{||A}} - \frac12 \tensor{\dot{\chi}}{^B^C_{||C}}\tensor{f}{_{||B}}\tensor{f}{_{||A}} + \frac{1}{12}\tensor{\ddot{\chi}}{_A_B}\tensor{f}{^{||B}}\tensor{f}{_{||C}}\tensor{f}{^{||C}} \, .
        \end{split} \nonumber
    \end{align}
\end{theorem}
\begin{proof}
    See \cite[Appendix C.5]{cjk}. Transformation of $\tensor{N}{_A}$ is correctly proven in \cite[Appendix~C]{cwy-transformations}.
\end{proof}

\subsection{Differential geometry on 2-sphere}

Finally, we list some of the important facts from differential geometry that shall later be used for our analysis of angular momentum.

\subsection{Curvature and covariant derivative}

\begin{lemma}
    On $S^2$, the curvature tensors have the following properties:
    \begin{enumerate}
        \item The Riemann tensor is proportional to the metric:
        \begin{equation}\label{eq:riemann}
            \tensor{R}{^A_B_C_D} = \tensor{\delta}{^A_C}\tensor{h}{_B_D} - \tensor{\delta}{^A_D}\tensor{h}{_B_C} ,
        \end{equation}
        \item The Ricci tensor is equal to the metric:
        \begin{equation*}
            \tensor{R}{_A_B} = \tensor{h}{_A_B} ,
        \end{equation*}
        \item The Ricci scalar is equal to $R = 2$.
    \end{enumerate}
\end{lemma}
\begin{proof}
    By explicit calculation of the Riemann tensor for $S^2$, one finds that it has only one nonzero independent component, say $R_{\theta\varphi\theta\varphi} = \sin^2\theta$. One can check that $\tensor{\delta}{^A_C}\tensor{h}{_B_D} - \tensor{\delta}{^A_D}\tensor{h}{_B_C}$ satisfies all of the symmetries of the Riemann tensor. Moreover:
    \begin{equation*}
        \tensor{\delta}{^\theta_\theta}\tensor{h}{_\varphi_\varphi} - \tensor{\delta}{^\theta_\varphi}\tensor{h}{_\varphi_\theta} = \sin^2\theta = R_{\theta\varphi\theta\varphi} = h_{A\theta}\tensor{R}{^A_\varphi_\theta_\varphi} = (h_{\theta\theta}\tensor{R}{^\theta_\varphi_\theta_\varphi} + h_{\varphi\theta}\tensor{R}{^\varphi_\varphi_\theta_\varphi}) = \tensor{R}{^\theta_\varphi_\theta_\varphi},
    \end{equation*}
    and the expression is zero for other independent components. Hence, it completely describes the Riemann tensor. The other identities follow easily from the first one.
\end{proof}

\begin{lemma}
    The Laplacian on $S^2$ satisfies the following commutation relation:
    \begin{equation}\label{eq:laplacian}
        D^2(\tensor{f}{_{||A}}) - \tensor{(D^2f)}{_{||A}} = \tensor{f}{_{||A}} ,
    \end{equation}
    where $f: S^2 \rightarrow \mathbb{R}$ is an arbitrary function.
\end{lemma}
\begin{proof}
    Note that for any one-form $\tensor{\alpha}{_A}$:
    \begin{equation*}
        \tensor{\alpha}{_A_{||CB}} = \tensor{\alpha}{_A_{||BC}} - \tensor{R}{^D_A_B_C}\tensor{\alpha}{_D} .
    \end{equation*}
    Therefore, using (\ref{eq:riemann}):
    \begin{equation*}
        \begin{split}
            D^2(\tensor{f}{_{||A}}) &= \tensor{f}{_{||A}^{||B}_{||B}} = \tensor{h}{^B^C}\tensor{f}{_{||ACB}} = \tensor{h}{^B^C}\tensor{f}{_{||CAB}} = \tensor{h}{^B^C} \left( \tensor{f}{_{||CBA}} - \tensor{R}{^D_C_B_A}\tensor{f}{_{||D}} \right) \\
            &= \tensor{h}{^B^C}\tensor{f}{_{||CBA}} - \tensor{h}{^B^C}\tensor{f}{_{||D}} \left( \tensor{\delta}{^D_B}\tensor{h}{_C_A} - \tensor{\delta}{^D_A}\tensor{h}{_C_B} \right) = \tensor{f}{^{||B}_{||BA}} - \tensor{f}{_{||A}} + 2\tensor{f}{_{||A}} \\
            &= \tensor{(D^2f)}{_{||A}} + \tensor{f}{_{||A}} .
        \end{split}
    \end{equation*}
\end{proof}

\subsection{Dipole-like functions}

\begin{lemma}\label{lem:dipole}
    Let $n: S^2 \rightarrow \mathbb{R}$ be a linear combination of $l = 1$ spherical harmonics. Then, it satisfies the following:
    \begin{align}
        (D^2 + 2)n &= 0 , \label{eq:dipole} \\
        (D^2 + 1) \tensor{n}{_{||A}} &= 0 \label{eq:dipole-gradient} , \\
        \tensor{n}{_{||AB}} &= -\tensor{h}{_A_B}n \label{eq:dipole-derivative} .
    \end{align}
\end{lemma}
\begin{proof}
    (\ref{eq:dipole}) follows directly from the fact that $n$ is a linear combination of $l = 1$ spherical harmonics. Using (\ref{eq:laplacian}) and (\ref{eq:dipole}), we have:
    \begin{equation*}
        D^2(\tensor{n}{_{||A}}) = \tensor{(D^2n)}{_{||A}} + \tensor{n}{_{||A}} = - 2\tensor{n}{_{||A}} + \tensor{n}{_{||A}} = - \tensor{n}{_{||A}} ,
    \end{equation*}
    which proves (\ref{eq:dipole-gradient}). Finally:
    \begin{equation*}
        \tensor{h}{^A^B}\tensor{n}{_{||AB}} = \tensor{n}{^{||A}_{||A}} = - 2n ,
    \end{equation*}
    which, multiplied by $\frac12 \tensor{h}{_A_B}$ on both sides, becomes (\ref{eq:dipole-derivative}).
\end{proof}

\subsection{Integral formulae}

\begin{lemma}\label{lem:full-divergence}
    Let $\tensor{X}{^A}$ be any vector field on $S^2$. Then:
    \begin{equation*}
        \int_{S^2} \tensor{X}{^A_{||A}} = 0.
    \end{equation*}
\end{lemma}
\begin{proof}
    This follows directly from Stokes' theorem written out in index notation (see \cite[Appendix E]{carroll2019} for details):
    \begin{equation*}
        \int_{S^2} \tensor{X}{^A_{||A}} = \int_{\partial S^2} n_AX^A,
    \end{equation*}
    where $n^A$ is the unit normal to the boundary, and the integral on the right-hand side is taken with respect to the corresponding volume form on the boundary. Since $\partial S^2 = \emptyset$, the integral is zero.
\end{proof}

\begin{lemma}[Integration by parts]\label{lem:parts}
    Let $T$ be a tensor field of type $(k,l)$ and $U$ a tensor field of type $(l,k-1)$, both on $S^2$. Then:
    \begin{equation*}
        \int_{S^2} \tensor{T}{^{AB_1 \dots B_{k-1}}_{C_1 \dots C_l}_{||A}}\tensor{U}{^{C_1 \dots C_l}_{B_1 \dots B_{k-1}}} = - \int_{S^2} \tensor{T}{^{AB_1 \dots B_{k-1}}_{C_1 \dots C_l}}\tensor{U}{^{C_1 \dots C_l}_{B_1 \dots B_{k-1}}_{||A}}.
    \end{equation*}
\end{lemma}
\begin{proof}
    Note that
    \begin{multline*}
        \left( \tensor{T}{^{AB_1 \dots B_{k-1}}_{C_1 \dots C_l}}\tensor{U}{^{C_1 \dots C_l}_{B_1 \dots B_{k-1}}} \right)_{||A} = \tensor{T}{^{AB_1 \dots B_{k-1}}_{C_1 \dots C_l}_{||A}}\tensor{U}{^{C_1 \dots C_l}_{B_1 \dots B_{k-1}}} \\ + \tensor{T}{^{AB_1 \dots B_{k-1}}_{C_1 \dots C_l}}\tensor{U}{^{C_1 \dots C_l}_{B_1 \dots B_{k-1}}_{||A}}.
    \end{multline*}
    The result then follows directly from Lemma \ref{lem:full-divergence}.
\end{proof}

Note that Lemma \ref{lem:parts} can be easily generalised to products of multiple components of tensors of suitable types, which will be frequently used throughout this paper.

\begin{lemma}\label{lem:zero-int}
    Let $u, v : S^2 \rightarrow \mathbb{R}$ such that $D^2u = - k(k + 1)u$ and $D^2v = - l(l + 1)v$, and let $n$ be as in Lemma \ref{lem:dipole}. Then, if $k \neq l$:
    \begin{equation*}
        \int_{S^2} n\tensor{\epsilon}{^A^B}\tensor{u}{_{||A}}\tensor{v}{_{||B}} = 0.
    \end{equation*}
\end{lemma}
\begin{proof}
    By integrating by parts, one obtains:
    \begin{equation*}
        \int_{S^2} n\epsilon^{AB}u_{||A}v_{||B} = - \int_{S^2} n_{||B}\epsilon^{AB}u_{||A}v.
    \end{equation*}
    The vector field $X^A = \epsilon^{AB}n_{||B}$ corresponds to rotations. Since, the Laplace operator is rotationally invariant, $X^A$ does not change its eigenspaces. 
    Thus, $X^Au_{||A} = \epsilon^{AB}n_{||B}u_{||A}$ belongs to the same eigenspace as $u$. However, unless $k = l$, $u$ and $v$ belong to different eigenspaces that are orthogonal to each other. Therefore, the considered integral, being in fact the inner product on the space of functions on $S^2$, is zero.
\end{proof}

\subsection{Tensor field decomposition}

\begin{lemma}\label{lem:potentials}
    Any symmetric traceless rank 2 tensor field $T$ on $S^2$ can be decomposed into two smooth functions $b$ and $c$ on $S^2$ or into one-form $v_A$, in the following way:
    \begin{align*}
        \tensor{T}{_A_B} &= v_{A||B} + v_{B||A} - h_{AB}\tensor{v}{^C_{||C}} \\
        \tensor{T}{_A_B} &= 2\tensor{c}{_{||AB}} - \tensor{h}{_A_B}D^2c + \tensor{\epsilon}{_A_C}\tensor{b}{^{||C}_{||B}} + \tensor{\epsilon}{_B_C}\tensor{b}{^{||C}_{||A}} \, ,
    \end{align*}
    where $b$ and $c$ have no $l = 0, 1$ parts. Moreover, the potentials are related by:
    \begin{equation*}
        v_A = c_{||A} + \epsilon_{AB}b^{||B} \, .
    \end{equation*}
\end{lemma}
\begin{proof}
    See \cite[Appendix E]{jezierski2002peeling}.
\end{proof}            
        
\section{Angular momenta at null infinity}

\subsection{The definitions}

Throughout this section, $n$ has the same meaning as in Lemma \ref{lem:dipole}.

\begin{definition}\label{def:momenta}
    We compare properties of three different notions of angular momentum, defined in terms of data on $\scri^+$:
    \begin{itemize}
        \item Chruściel-Jezierski-Kijowski \cite{cjk} (or rather Dray-Streubel \cite{DS}) angular momentum:
        \begin{equation}\label{eq:cjk}
            \cjk = -\frac{1}{8\pi} \int_{S^2} \left( 3\tensor{N}{_A} + \frac14 \tensor{\chi}{_A_B}\tensor{\chi}{^B^C_{||C}} \right) \tensor{\epsilon}{^A^D}\tensor{n}{_{||D}},
        \end{equation}
        \item Compere-Nichols angular momentum \cite{cn}:
        \begin{equation}\label{eq:cn}
            \cn = \cjk - \frac{\alpha - 1}{32\pi} \int_{S^2} \tensor{\epsilon}{^A^D}\tensor{n}{_{||D}}\tensor{\chi}{_A_B}\tensor{\chi}{^B^C_{||C}},
        \end{equation}
        where $\alpha \in \mathbb{R}$ is an arbitrary free parameter,
        \item Chen-Wang-Yau angular momentum \cite{cwy}:
        \begin{equation}\label{eq:cwy}
            \cwy = \cjk + \frac{1}{4\pi} \int_{S^2} m\tensor{\epsilon}{^A^B}\tensor{n}{_{||B}}\tensor{c}{_{||A}},
        \end{equation}
        where $c$ is the unique solution to
        \begin{equation}\label{eq:c-solution}
            D^2(D^2 + 2)c = \tensor{\chi}{_A_B^{||BA}}
        \end{equation}
        such that $c$ has no $l = 0,1$ parts.
    \end{itemize}
\end{definition}

At this point, one may already notice that $\cn = \cjk$ for $\alpha = 1$. Before we move any further, it is useful to note that, by Lemma \ref{lem:potentials}, the definitions may be rewritten in terms of potentials.

\begin{lemma}\label{lem:momenta-potentials}
    The shear tensor and the angular momenta from Definition \ref{def:momenta} take the following forms when rewritten in terms of smooth functions $b, c: S^2 \rightarrow \mathbb{R}$
    \begin{align}
        \tensor{\chi}{_A_B} &= 2\tensor{c}{_{||AB}} - \tensor{h}{_A_B}D^2c + \tensor{\epsilon}{_A_C}\tensor{b}{^{||C}_{||B}} + \tensor{\epsilon}{_B_C}\tensor{b}{^{||C}_{||A}}, \label{eq:chi} \\
        \cjk &= \frac{1}{8\pi} \int_{S^2} \left( - 3N_A\epsilon^{AD}n_{||D} - 2n_{||A}c^{||A}D^2b + \frac12nD^2bD^2c - \frac12n_{||A}D^2b(D^2c)^{||A} \right), \label{eq:cjk-pot} \\
        \cn &= \frac{1}{8\pi} \int_{S^2} \left( - 3N_A\epsilon^{AD}n_{||D} - 2\alpha n_{||A}c^{||A}D^2b + \frac{\alpha}{2}nD^2bD^2c - \frac{\alpha}{2}n_{||A}D^2b(D^2c)^{||A} \right), \label{eq:cn-pot} \\
        \begin{split}
            \cwy &= \frac{1}{8\pi} \int_{S^2} \left( - 3N_A\epsilon^{AD}n_{||D} - 2n_{||A}c^{||A}D^2b + \frac12nD^2bD^2c - \frac12n_{||A}D^2b(D^2c)^{||A} \right. \\
            &\phantom{=}\, \left. + 2m\epsilon^{AD}n_{||D}c_{||A} \vphantom{\frac12}\right).
        \end{split} \label{eq:cwy-pot}
    \end{align}
\end{lemma}
\begin{proof}
    By Lemma \ref{lem:potentials} we have (\ref{eq:chi}). Note that $c$ in (\ref{eq:cwy}) is the same as in (\ref{eq:chi}) - this can be seen by substituting (\ref{eq:chi}) into (\ref{eq:c-solution}). The above formulas are then obtained by direct substitution. It suffices to show that terms quadratic in $b$ and terms quadratic in $c$ evaluate to 0, which is done by integrating the terms by parts and obtaining expressions that are covariant derivatives of some function with respect to the index $A$. For details see Appendix~\ref{app:potentials}.
\end{proof}

By looking at the formulas written out in terms of potentials, it is easier to notice how exactly they differ. In particular, one may notice that in certain special cases, all of these definitions are equivalent.

\begin{theorem}\label{thm:cjk-cn}
    If $b = 0$ in (\ref{eq:chi}), then $\cjk = \cn$.
\end{theorem}
\begin{proof}
    It suffices to substitute $b = 0$ in (\ref{eq:cjk-pot}) and (\ref{eq:cn-pot}).
\end{proof}

\begin{theorem}
    If $\tensor{m}{_{||A}} = 0$ for all $A$, then $\cjk = \cwy$.
\end{theorem}
\begin{proof}
    Since $\tensor{m}{_{||A}} = 0$, $m$ is constant on $S^2$ and may be taken out of the integral. By Lemma~\ref{lem:potentials} $c$ has no $l = 0, 1$ parts. On the other hand $D^2n = -2n$. Therefore, the statement follows from Lemma~\ref{lem:zero-int}.
\end{proof}

\subsection{Time evolution}

We now turn to the fluxes of the presented angular momenta. Firstly, we present them for $\cjk$ and $\cn$, as they may be expressed directly in terms of the data on $\scri^+$.

\begin{theorem}
    $\cjk$ and $\cn$ obey the following evolution equations:
    \begin{align}
        \dcjk &= \frac{1}{8\pi} \int_{S^2} \tensor{\epsilon}{^A^D}\tensor{\dot{\chi}}{^B^C} \left( \frac12 n\tensor{h}{_C_D}\tensor{\chi}{_A_B} - \frac14 \tensor{n}{_{||D}}\tensor{\chi}{_A_B_{||C}} - \frac14 \tensor{n}{_{||D}}\tensor{h}{_A_B}\tensor{\chi}{_{CE}^{||E}} \right), \label{eq:dcjk-chi} \\
        \dcn &= \frac{1}{32\pi} \int_{S^2} \tensor{\epsilon}{^A^D}\tensor{\dot{\chi}}{^B^C} \left( (3 - \alpha) n\tensor{h}{_C_D}\tensor{\chi}{_A_B} + (\alpha - 2) \tensor{n}{_{||D}}\tensor{\chi}{_A_B_{||C}} - \alpha \tensor{n}{_{||D}}\tensor{h}{_A_B}\tensor{\chi}{_{CE}^{||E}} \right). \label{eq:dcn-chi}
    \end{align}
    In particular, $\dcjk = \dcn = 0$ if $\dot{\chi} = 0$.
\end{theorem}
\begin{proof}
    The statement follows directly from explicit calculations. See Appendix \ref{app:flux-chi} for details.
\end{proof}

So far, no additional assumptions about the physical system have been made. However, in order to simplify the formulas and to draw any meaningful conclusions, we shall assume $b = 0$ when dealing with fluxes.

\begin{theorem}\label{thm:fluxes}
    If $b = 0$ in (\ref{eq:chi}), then the angular momenta from Definition \ref{def:momenta} obey the following evolution equations:
    \begin{align*}
        \dcjk &= - \frac{1}{16\pi} \int_{S^2} \tensor{\epsilon}{^A^D}\tensor{n}{_{||D}}c_{||A} D^2(D^2 + 2)\dot{c}, \\
        \dcn &= - \frac{1}{16\pi} \int_{S^2} \tensor{\epsilon}{^A^D}\tensor{n}{_{||D}}c_{||A} D^2(D^2 + 2)\dot{c}, \\
        \dcwy &= \frac{1}{8\pi} \int_{S^2} \tensor{\epsilon}{^A^D}\tensor{n}{_{||D}} \left( 2m\tensor{\dot{c}}{_{||A}} - \tensor{c}{_{||A}}\tensor{\dot{c}}{^{||BC}}\tensor{\dot{c}}{_{||BC}} + \frac12\tensor{c}{_{||A}}(D^2\dot{c})^2 \right)
        \, .
    \end{align*}
    In particular, $\dcjk = \dcn = \dcwy = 0$ when $\dot{c} = 0$, which, in this case, is equivalent to $\dot{\chi} = 0$.
\end{theorem}
\begin{proof}
    It suffices to substitute (\ref{eq:chi}) in (\ref{eq:dcjk-chi}) and (\ref{eq:dcn-chi}). For $\dcwy$, the additional term is differentiated explicitly. For details see Appendix \ref{app:flux-pot}.
\end{proof}

Note that this result agrees with Theorem \ref{thm:cjk-cn} - for $b = 0$, $\cjk$ and $\cn$ are equal. $\dcwy$, however, has additional terms, which originate from differentiating the correction term in (\ref{eq:cwy}). Nonetheless, all fluxes behave reasonably, as for $\dot{\chi} = 0$, no angular momentum should be radiated out of the system.

\subsection{Transformations under supertranslations}

Here, we present how different definitions of angular momentum transform under supertranslations. Similarly to the previous section, an additional assumption has to be made in order to obtain sufficiently clean expressions for comparison.

\begin{theorem}\label{thm:supertranslations}
    Suppose that $\dot{\chi} = 0$. Then, under a supertranslation $u \mapsto u - f(\tensor{x}{^A})$, the angular momenta from Definition \ref{def:momenta} obey the following transformation rules:
    \begin{align} \label{eq:trans-cwy}
        \cjk &\longmapsto \cjk + \frac{1}{8\pi} \int_{S^2} f \left( 3{\epsilon}^{DA}\tensor{n}{_{||D}}\tensor{m}{_{||A}} + \frac14n_{||A}D^AD^2(D^2 + 2)b \right), \nonumber \\
        \begin{split}
            \cn &\longmapsto \cn + \frac{1}{8\pi} \int_{S^2} f \left( 3{\epsilon}^{DA}\tensor{n}{_{||D}}\tensor{m}{_{||A}} + \frac32(\alpha - 1)nD^2(D^2+2)b \right. \\
            &\phantom{=}\, \left. + \frac14(3 - 2\alpha)n_{||A}D^AD^2(D^2 + 2)b \right),
        \end{split} \\
        \cwy &\longmapsto \cwy + \frac{1}{8\pi} \int_{S^2} f \left({\epsilon}^{DA}\tensor{n}{_{||D}}\tensor{m}{_{||A}} + \frac14n_{||A}D^AD^2(D^2 + 2)b \right)\, . \nonumber
    \end{align}
\end{theorem}
\begin{proof}
    The above transformation rules are obtained by substituting formulas from Theorem~\ref{thm:data-transformations} into \cref{eq:cjk-pot,eq:cn-pot,eq:cwy-pot} and setting $\dot{\chi} = 0$. Then, by integrating by parts, one may see that some terms vanish and the rest can be manipulated into the form above. For details, see Appendix \ref{app:transformations}.
\end{proof}

From the above result, it is clear that $\cjk$ and $\cwy$ have the simplest transformation formula. The introduction of additional terms with free parameter $\alpha$ in $\cn$ comes with the disadvantage of a more complex expression. On the other hand, the correction term in $\cwy$ does not produce any additional terms in the transformation formula, as compared with $\cjk$. In fact, this correction term partially cancels the term independent of $b$ in the transformation expression. It is, therefore, easy to notice that, by properly changing the coefficient of this correction term, one may obtain a definition of angular momentum that is conveniently close to being invariant. Hence, in the following section, the new definition is presented.

\subsection{The proposed definition of angular momentum}

\begin{definition}
    We propose a new definition of angular momentum:
    \begin{equation}\label{eq:j}
        J = \cjk + \frac{3}{8\pi} \int_{S^2} m\tensor{\epsilon}{^A^D}\tensor{n}{_{||D}}\tensor{c}{_{||A}}.
    \end{equation}
\end{definition}

\begin{theorem}\label{thm:j}
    The new angular momentum $J$ proposed in (\ref{eq:j}) has the following properties:
    \begin{itemize}
        \item Under the assumption that $b = 0$, $J$ obeys the following evolution equation:
        \begin{equation*}
            \dot{J} = \frac{1}{8\pi} \int_{S^2} \tensor{\epsilon}{^A^D}\tensor{n}{_{||D}} \left( 3m\tensor{\dot{c}}{_{||A}} - \frac32\tensor{c}{_{||A}}\tensor{\dot{c}}{^{||BC}}\tensor{\dot{c}}{_{||BC}} + \frac34\tensor{c}{_{||A}}(D^2\dot{c})^2 + \frac14\tensor{c}{_{||A}}D^2(D^2 + 2)\dot{c} \right),
        \end{equation*}
        \item If $\dot{\chi} = 0$, then, under a supertranslation $(u, \tensor{x}{^A}) \mapsto (u - f(\tensor{x}{^A}), \tensor{x}{^A})$, $J$ obeys the following transformation rule:
        \begin{equation*}
            J \longmapsto J + \frac{1}{32\pi} \int_{S^2} fn_{||A}D^AD^2(D^2 + 2)b.
        \end{equation*}
    \end{itemize}
\end{theorem}
\begin{proof}
    The flux is calculated explicitly. See Appendix \ref{app:flux-pot} for details. The transformation rule is easily deduced by seeing that the first term in (\ref{eq:trans-cwy}) vanishes due to the changed coefficient in the correction term, as stated earlier.
\end{proof}

\begin{theorem}\label{thm:j-transformation}
    Suppose that $\dot{\chi} = 0$ and $n_{||A}D^AD^2(D^2 + 2)b$ is orthogonal to $f$. Then, $J$ is invariant under the supertranslation $(u, \tensor{x}{^A}) \mapsto (u - f(\tensor{x}{^A}), \tensor{x}{^A})$:
    \begin{equation*}
        J \longmapsto J.
    \end{equation*}
\end{theorem}
\begin{proof}
    Follows directly from Theorem \ref{thm:j}.
\end{proof}

As a result, one obtains another definition of angular momentum at null infinity which obeys a much simpler transformation rule under supertranslations, as compared to the other definitions.

\subsection{Transformations with inclusion of the superenergy terms}

So far, we have considered only transformations of the data on $\scri^+$, which gave the transformation formulae for angular momenta presented above. However, one could as well consider how the rotational vector field itself is transformed under supertranslations, which introduces additional terms into the transformation rules. To see this, we first note:

\begin{theorem}\label{thm:field-transformation}
    For any function $f: S^2 \rightarrow \mathbb{R}$, a rotational vector field $X = X^A\partial_{x^A}$ on $\scri^+$ obeys the following transformation rule under the supertranslation $\Bar{u} = u - f(x^A)$, $\Bar{x}^A = x^A$:
    \begin{equation*}
        \Bar{X} = X^A \partial_{\Bar{x}^A} - X(f) \partial_{\Bar{u}}.
    \end{equation*}
\end{theorem}
\begin{proof}
    This is easily proven by noticing that
    \begin{equation*}
        \partial_{x^A} = \frac{\partial\Bar{x}^A}{\partial x^A}\partial_{\Bar{x}^A} + \frac{\partial\Bar{u}}{\partial x^A}\partial_{\Bar{u}} = \partial_{\Bar{x}^A} - \frac{\partial f}{\partial x^A}\partial_{\Bar{u}},
    \end{equation*}
    and hence,
    \begin{equation*}
        \Bar{X} = X^A \partial_{\Bar{x}^A} - X^A \frac{\partial f}{\partial x^A} \partial_{\Bar{u}} = X^A \partial_{\Bar{x}^A} - X(f) \partial_{\Bar{u}}.
    \end{equation*}
\end{proof}

In order to properly transform $\cjk$ with the rotational vector field taken into consideration, we shall look back at the defining Hamiltonian. In practice, we transform the formula \cite[(C.100)]{cjk} by applying Theorem~\ref{thm:data-transformations} and Theorem~\ref{thm:field-transformation}. As a result, we get:

\begin{theorem}
    Suppose that $\dot{\chi} = 0$. Then, under a supertranslation $u \mapsto u - f(\tensor{x}{^A})$, the angular momenta from Definition \ref{def:momenta} obey the following transformation rules:
    \begin{align*}
        \cjk &\longmapsto \cjk + \frac{1}{8\pi} \int_{S^2} f \left( {\epsilon}^{DA}\tensor{n}{_{||D}}\tensor{m}{_{||A}} + \frac14\epsilon^{AD}n_{||D}(D^2D^2 - 1)(v_A - 3c_{||A}) \right), \nonumber \\
        \begin{split}
            \cn &\longmapsto \cn + \frac{1}{8\pi} \int_{S^2} f \left( {\epsilon}^{DA}\tensor{n}{_{||D}}\tensor{m}{_{||A}} + \frac32(\alpha - 1)nD^2(D^2+2)b \right. \\
            &\phantom{=}\, \left. + \frac14(3 - 2\alpha)n_{||A}D^AD^2(D^2 + 2)b - \frac12\epsilon^{AD}n_{||D}D_AD^2(D^2 + 2)c \right),
        \end{split} \nonumber \\
        \cwy &\longmapsto \cwy + \frac{1}{8\pi} \int_{S^2} f \left({\epsilon}^{AD}\tensor{n}{_{||D}}\tensor{m}{_{||A}} + \frac14\epsilon^{AD}n_{||D}(D^2D^2 - 1)(v_A - 3c_{||A}) \right)\, .
    \end{align*}
\end{theorem}
\begin{proof}
    This follows from Theorem \ref{thm:supertranslations} with the inclusion of additional term from \cite[(C.100)]{cjk}. The only nonzero term which was not considered earlier is
    \begin{equation*}
        \frac{1}{16\pi}\int_{S^2}\left(4m-\tensor{\chi}{^B^C_{||BC}}\right)X^u.
    \end{equation*}
    After transforming it under a supertranslation $u \mapsto u - f(\tensor{x}{^A})$, we get: (the second equality follows from integration by parts)
    \begin{equation*}
    \begin{split}
        &\phantom{=} \frac{1}{16\pi}\int_{S^2}\left(4m-\tensor{\chi}{^B^C_{||BC}}+2\tensor{f}{^{||BC}_{||BC}}-D^2D^2f\right)\Bar{X}^{\Bar{u}} \\
        &= - \frac{1}{16\pi}\int_{S^2}\left(4m-\tensor{\chi}{^B^C_{||BC}}+D^2(D^2+2)f\right)\epsilon^{AD}n_{||D}f_{||A} \\
        &= - \frac{1}{8\pi}\int_{S^2}\left(2m\epsilon^{AD}n_{||D}f_{||A} - \frac12\tensor{\chi}{^B^C_{||BC}}\epsilon^{AD}n_{||D}f_{||A}\right) \\
        &= \frac{1}{8\pi}\int_{S^2}\left(2m_{||A}\epsilon^{AD}n_{||D}f + \frac12\epsilon^{AD}n_{||D}f_{||A}D^2(D^2+2)c\right) \\
        &= \frac{1}{8\pi}\int_{S^2}f\left(2m_{||A}\epsilon^{AD}n_{||D}-\frac12\epsilon^{AD}n_{||D}D_AD^2(D^2+2)c\right).
    \end{split}
    \end{equation*}
    By adding the above terms to the previously derived formulae, and applying Lemma \ref{lem:potentials}, we obtain the presented transformation rules.
\end{proof}

Considering the new formulae, one may propose a new definition:
\begin{definition}
    \begin{equation}\label{eq:j-tilde}
        \Tilde{J} = \cjk + \frac{1}{8\pi} \int_{S^2} m\tensor{\epsilon}{^A^D}\tensor{n}{_{||D}}\tensor{c}{_{||A}},
    \end{equation}
\end{definition}
which has the following properties:
\begin{theorem}\label{thm:j-tilde}
    The angular momentum $\Tilde{J}$ defined in (\ref{eq:j-tilde}) has the following properties:
    \begin{itemize}
        \item Under the assumption that $b = 0$, $\Tilde{J}$ obeys the following evolution equation:
        \begin{equation*}
            \dot{\Tilde{J}} = \frac{1}{8\pi} \int_{S^2} \tensor{\epsilon}{^A^D}\tensor{n}{_{||D}} \left( m\tensor{\dot{c}}{_{||A}} - \frac12\tensor{c}{_{||A}}\tensor{\dot{c}}{^{||BC}}\tensor{\dot{c}}{_{||BC}} + \frac14\tensor{c}{_{||A}}(D^2\dot{c})^2 - \frac14\tensor{c}{_{||A}}D^2(D^2 + 2)\dot{c} \right),
        \end{equation*}
        \item If $\dot{\chi} = 0$, then, under a supertranslation $(u, \tensor{x}{^A}) \mapsto (u - f(\tensor{x}{^A}), \tensor{x}{^A})$, $\Tilde{J}$ obeys the following transformation rule:
        \begin{equation*}
            \Tilde{J} \longmapsto \Tilde{J} + \frac{1}{32\pi} \int_{S^2} f\epsilon^{AD}n_{||D}(D^2D^2 - 1)(v_A - 3c_{||A}).
        \end{equation*}
    \end{itemize}
\end{theorem}
\begin{proof}
    This result is proved analogously to Theorem~\ref{thm:j} but with inclusion of additional transformation terms as shown above.
\end{proof}

Accordingly, we may note an analogous result to Theorem \ref{thm:j-transformation}.

\begin{theorem}
    Suppose that $\chi$ is such that $\dot{\chi} = 0$ and $\epsilon^{AD}n_{||D}(D^2D^2 - 1)(v_A - 3c_{||A})$ is orthogonal to $f$. Then, $\Tilde{J}$ is invariant under the supertranslation $(u, \tensor{x}{^A}) \mapsto (u - f(\tensor{x}{^A}), \tensor{x}{^A})$:
    \begin{equation*}
        \Tilde{J} \longmapsto \Tilde{J}.
    \end{equation*}
\end{theorem}
\begin{proof}
    Follows directly from Theorem \ref{thm:j-tilde}.
\end{proof}
        
\section{Conclusion}

This article aimed to present different notions of angular momentum at null infinity found in the literature and compare their selected properties. To aid this goal, essential definitions and theorems were first reviewed. Then, the chosen angular momenta were presented, in two different forms. Consequently, the formulas for fluxes and transformations under supertranslations of those momenta were calculated and compared. Finally, given the computed transformation rules, new definitions of angular momentum were proposed, and their core properties were presented for comparison with other definitions.

Ultimately, it is worth noting that the proposed definitions of angular momentum, like $\cwy$, are described in terms of the potential $c$. However, one may modify these definitions in such a way to emphasise its geometric meaning. Namely, $c_{||A}$ may be replaced by ${v_A = c_{||A} + \epsilon_{AB}b^{||B}}$, a first order potential for $\chi$, meaning that ${\chi_{AB} = v_{A||B} + v_{B||A} - h_{AB}\tensor{v}{^C_{||C}}}$ \cite[Appendix E]{jezierski2002peeling}. Since $b$ is invariant under supertranslations, the properties stated in Theorem~\ref{thm:j} and Theorem~\ref{thm:j-tilde} would not change due to such a modification. On the other hand, one would obtain a fully geometrical definition, which may be looked at in a more general context.
        
\bibliographystyle{plain}
\bibliography{bibliography}
\appendix

\raggedbottom
\clearpage

\section{Angular momenta in terms of potentials}\label{app:potentials}

Since the $\tensor{N}{_A}$ terms are left unchanged in Lemma \ref{lem:momenta-potentials}, it may be ommited from calculations. We start by substituting \eqref{eq:chi} into \eqref{eq:cjk}. For clarity, we simplify the integrand first.

\begin{equation*}
\begin{split}
    - \frac14\tensor{\chi}{_A_B}\tensor{\chi}{^B^C_{||C}} &= - \frac14 \left( 2\tensor{c}{_{||AB}} - \tensor{h}{_A_B}D^2c + \tensor{\epsilon}{_A_F}\tensor{b}{^{||F}_{||B}} + \tensor{\epsilon}{_B_F}\tensor{b}{^{||F}_{||A}} \right) \\
    &\phantom{=}\ \times \left( 2\tensor{c}{^{||BC}} - \tensor{h}{^B^C}D^2c + \tensor{\epsilon}{^B^E}\tensor{b}{_{||E}^{||C}} + \tensor{\epsilon}{^C^E}\tensor{b}{_{||E}^{||B}} \right)_{||C} \\
    &= - \tensor{c}{_{||AB}}\tensor{c}{^{||BC}_{||C}} + \frac12\tensor{c}{_{||AB}}\tensor{(D^2c)}{^{||B}} - \frac12\tensor{\epsilon}{^B^E}\tensor{c}{_{||AB}}\tensor{b}{_{||E}^{||C}_{||C}} - \frac12\tensor{c}{_{||AB}}\tensor{\epsilon}{^C^E}\tensor{b}{_{||E}^{||B}_{||C}} \\
    &\phantom{=}\ + \frac12\tensor{c}{_{||A}^{||C}_{||C}}D^2c - \frac14\tensor{(D^2c)}{_{||A}}D^2c + \frac14\tensor{\epsilon}{_A_E}\tensor{b}{^{||EC}_{||C}}D^2c + \frac14\tensor{\epsilon}{^C^E}\tensor{b}{_{||EAC}}D^2c \\
    &\phantom{=}\ - \frac12\tensor{\epsilon}{_A_F}\tensor{b}{^{||F}_{||B}}\tensor{c}{^{||BC}_{||C}} + \frac14\tensor{\epsilon}{_A_F}\tensor{b}{^{||F}_{||B}}\tensor{(D^2c)}{^{||B}} - \frac14\tensor{\epsilon}{_A_F}\tensor{\epsilon}{^B^E}\tensor{b}{^{||F}_{||B}}\tensor{b}{_{||E}^{||C}_{||C}} \\
    &\phantom{=}\ - \frac14\tensor{\epsilon}{_A_F}\tensor{\epsilon}{^C^E}\tensor{b}{^{||F}_{||B}}\tensor{b}{_{||E}^{||B}_{||C}} - \frac12\tensor{\epsilon}{_B_F}\tensor{b}{^{||F}_{||A}}\tensor{c}{^{||BC}_{||C}} + \frac14\tensor{\epsilon}{_B_F}\tensor{b}{^{||F}_{||A}}\tensor{(D^2c)}{^{||B}} \\
    &\phantom{=}\ - \frac14\tensor{b}{^{||E}_{||A}}\tensor{b}{_{||E}^{||C}_{||C}} - \frac14\tensor{\epsilon}{_B_F}\tensor{\epsilon}{^C^E}\tensor{b}{^{||F}_{||A}}\tensor{b}{_{||E}^{||B}_{||C}} \\
    &= - \tensor{c}{_{||AB}}\tensor{(D^2c)}{^{||B}} - \tensor{c}{_{||AB}}\tensor{c}{^{||B}} + \frac12\tensor{c}{_{||AB}}\tensor{(D^2c)}{^{||B}} - \frac12\tensor{\epsilon}{^B^E}\tensor{c}{_{||AB}}\tensor{(D^2b)}{_{||E}} \\
    &\phantom{=}\ - \frac12\tensor{\epsilon}{^B^E}\tensor{c}{_{||AB}}\tensor{b}{_{||E}} - \frac12\tensor{\epsilon}{^C^E}\tensor{c}{_{||AB}}\tensor{b}{_{||E}^{||B}_{||C}} + \frac12\tensor{(D^2c)}{_{||A}}D^2c + \frac12\tensor{c}{_{||A}}D^2c \\
    &\phantom{=}\ - \frac14\tensor{(D^2c)}{_{||A}}D^2c + \frac14\tensor{\epsilon}{_A_E}\tensor{(D^2b)}{^{||E}}D^2c + \frac14\tensor{\epsilon}{_A_E}\tensor{b}{^{||E}}D^2c + \frac14\tensor{\epsilon}{^C^E}\tensor{b}{_{||EAC}}D^2c \\
    &\phantom{=}\ - \frac12\tensor{\epsilon}{_A_F}\tensor{b}{^{||F}_{||B}}\tensor{(D^2c)}{^{||B}} - \frac12\tensor{\epsilon}{_A_F}\tensor{b}{^{||F}_{||B}}\tensor{c}{^{||B}} + \frac14\tensor{\epsilon}{_A_F}\tensor{b}{^{||F}_{||B}}\tensor{(D^2c)}{^{||B}} \\
    &\phantom{=}\ - \frac14\left(\tensor{\delta}{_A^B}\tensor{\delta}{_F^E} - \tensor{\delta}{_A^E}\tensor{\delta}{_F^B}\right)\tensor{b}{^{||F}_{||B}}\tensor{(D^2b)}{_{||E}} \\
    &\phantom{=}\ - \frac14\left(\tensor{\delta}{_A^B}\tensor{\delta}{_F^E} - \tensor{\delta}{_A^E}\tensor{\delta}{_F^B}\right)\tensor{b}{^{||F}_{||B}}\tensor{b}{_{||E}} \\
    &\phantom{=}\ - \frac14\left(\tensor{\delta}{_A^C}\tensor{\delta}{_F^E} - \tensor{\delta}{_A^E}\tensor{\delta}{_F^C}\right)\tensor{b}{^{||F}_{||B}}\tensor{b}{_{||E}^{||B}_{||C}} - \frac12\tensor{\epsilon}{_B_F}\tensor{b}{^{||F}_{||A}}\tensor{(D^2c)}{^{||B}} \\
    &\phantom{=}\ - \frac12\tensor{\epsilon}{_B_F}\tensor{b}{^{||F}_{||A}}\tensor{c}{^{||B}} + \frac14\tensor{\epsilon}{_B_F}\tensor{b}{^{||F}_{||A}}\tensor{(D^2c)}{^{||B}} - \frac14\tensor{b}{^{||E}_{||A}}\tensor{(D^2b)}{_{||E}} - \frac14\tensor{b}{^{||E}_{||A}}\tensor{b}{_{||E}} \\
    &\phantom{=}\ - \frac14\left(\tensor{\delta}{_B^C}\tensor{\delta}{_F^E} - \tensor{\delta}{_B^E}\tensor{\delta}{_F^C}\right)\tensor{b}{^{||F}_{||A}}\tensor{b}{_{||E}^{||B}_{||C}}
\end{split}
\end{equation*}
\begin{equation*}
\begin{split}
    &= - \frac12\tensor{c}{_{||AB}}\tensor{(D^2c)}{^{||B}} - \tensor{c}{_{||AB}}\tensor{c}{^{||B}} - \frac12\tensor{\epsilon}{^B^E}\tensor{c}{_{||AB}}\tensor{(D^2b)}{_{||E}} - \frac12\tensor{\epsilon}{^B^E}\tensor{c}{_{||AB}}\tensor{b}{_{||E}} - \frac12\tensor{\epsilon}{^C^E}\tensor{c}{_{||AB}}\tensor{b}{_{||E}^{||B}_{||C}} \\
    &\phantom{=}\ + \frac14\tensor{(D^2c)}{_{||A}}D^2c + \frac12\tensor{c}{_{||A}}D^2c + \frac14\tensor{\epsilon}{_A_E}\tensor{(D^2b)}{^{||E}}D^2c + \frac14\tensor{\epsilon}{_A_E}\tensor{b}{^{||E}}D^2c + \frac14\tensor{\epsilon}{^C^E}\tensor{b}{_{||EAC}}D^2c \\
    &\phantom{=}\ - \frac14\tensor{\epsilon}{_A_F}\tensor{b}{^{||F}_{||B}}\tensor{(D^2c)}{^{||B}} - \frac12\tensor{\epsilon}{_A_F}\tensor{b}{^{||F}_{||B}}\tensor{c}{^{||B}} - \frac14\tensor{b}{_{||EA}}\tensor{(D^2b)}{^{||E}} + \frac14\tensor{(D^2b)}{_{||A}}D^2b - \frac14\tensor{b}{_{||EA}}\tensor{b}{^{||E}} \\
    &\phantom{=}\ + \frac14\tensor{b}{_{||A}}D^2b - \frac14\tensor{b}{^{||E}_{||B}}\tensor{b}{_{||E}^{||B}_{||A}} + \frac14\tensor{b}{^{||C}_{||B}}\tensor{b}{_{||A}^{||B}_{||C}} - \frac14\tensor{\epsilon}{_B_F}\tensor{b}{^{||F}_{||A}}\tensor{(D^2b)}{^{||B}} - \frac12\tensor{\epsilon}{_B_F}\tensor{b}{^{||F}_{||A}}\tensor{c}{^{||B}} \\
    &\phantom{=}\ - \frac14\tensor{b}{^{||E}_{||A}}\tensor{(D^2b)}{_{||E}} - \frac14\tensor{b}{^{||E}_{||A}}\tensor{b}{_{||E}} - \frac14\tensor{b}{^{||E}_{||A}}\tensor{(D^2b)}{_{||E}} - \frac14\tensor{b}{^{||E}_{||A}}\tensor{b}{_{||E}} + \frac14\tensor{b}{^{||C}_{||A}}\tensor{(D^2b)}{_{||C}} \\
    &= - \frac12\tensor{c}{_{||AB}}\tensor{(D^2c)}{^{||B}} - \tensor{c}{_{||AB}}\tensor{c}{^{||B}} - \frac12\tensor{\epsilon}{^B^E}\tensor{c}{_{||AB}}\tensor{(D^2b)}{_{||E}} - \frac12\tensor{\epsilon}{^B^E}\tensor{c}{_{||AB}}\tensor{b}{_{||E}} \\
    &\phantom{=}\ + \frac12\tensor{\epsilon}{^C^E}\tensor{c}{_{||A}^{||B}}\tensor{R}{^F_E_C_B}\tensor{b}{_{||F}} + \frac14\tensor{(D^2c)}{_{||A}}D^2c + \frac12\tensor{c}{_{||A}}D^2c + \frac14\tensor{\epsilon}{_A_E}\tensor{(D^2b)}{^{||E}}D^2c + \frac14\tensor{\epsilon}{_A_E}\tensor{b}{^{||E}}D^2c \\
    &\phantom{=}\ - \frac14\tensor{\epsilon}{^C^E}\tensor{R}{^F_E_C_A}\tensor{b}{_{||F}}D^2c - \frac14\tensor{\epsilon}{_A_F}\tensor{b}{^{||F}_{||B}}\tensor{(D^2b)}{^{||B}} - \frac12\tensor{\epsilon}{_A_F}\tensor{b}{^{||F}_{||B}}\tensor{c}{^{||B}} - \frac14\tensor{\epsilon}{_B_F}\tensor{b}{^{||F}_{||A}}\tensor{(D^2c)}{^{||B}} \\
    &\phantom{=}\ - \frac12\tensor{\epsilon}{_B_F}\tensor{b}{^{||F}_{||A}}\tensor{c}{^{||B}} - \frac12\tensor{b}{_{||AB}}\tensor{(D^2b)}{^{||B}} + \frac14\tensor{(D^2b)}{_{||A}}D^2b  - \frac34\tensor{b}{_{||AB}}\tensor{b}{^{||B}} + \frac14\tensor{b}{_{||A}}D^2b \\
    &\phantom{=}\ - \frac14\tensor{b}{^{||C}_{||B}}\tensor{b}{_{||C}^{||B}_{||A}} + \frac14\tensor{b}{^{||C}_{||B}}\tensor{b}{_{||A}^{||B}_{||C}} \\
    &= - \frac12\tensor{c}{_{||AB}}\tensor{(D^2c)}{^{||B}} - \tensor{c}{_{||AB}}\tensor{c}{^{||B}} + \frac14\tensor{(D^2c)}{_{||A}}D^2c + \frac12\tensor{c}{_{||A}}D^2c - \frac12\tensor{\epsilon}{^B^C}\tensor{c}{_{||AB}}\tensor{(D^2b)}{_{||C}} \\
    &\phantom{=}\ - \frac12\tensor{\epsilon}{^B^C}\tensor{c}{_{||AB}}\tensor{b}{_{||C}} + \frac12\tensor{\epsilon}{^C^E}\tensor{c}{_{||A}^{||B}}\tensor{b}{_{||F}}\left(\tensor{\delta}{^F_C}\tensor{h}{_E_B} - \tensor{\delta}{^F_B}\tensor{h}{_E_C}\right) + \frac14\tensor{\epsilon}{_A_B}\tensor{(D^2b)}{^{||B}}D^2c \\
    &\phantom{=}\ + \frac14\tensor{\epsilon}{_A_B}\tensor{b}{^{||B}}D^2c - \frac14\tensor{\epsilon}{^C^E}\tensor{b}{_{||F}}D^2c\left(\tensor{\delta}{^F_C}\tensor{h}{_E_A} - \tensor{\delta}{^F_A}\tensor{h}{_E_C}\right) - \frac14\tensor{\epsilon}{_A_C}\tensor{b}{^{||C}_{||B}}\tensor{(D^2c)}{^{||B}} \\
    &\phantom{=}\ - \frac12\tensor{\epsilon}{_A_C}\tensor{b}{^{||C}_{||B}}\tensor{c}{^{||B}} - \frac14\tensor{\epsilon}{_B_C}\tensor{b}{^{||C}_{||A}}\tensor{(D^2c)}{^{||B}} - \frac12\tensor{\epsilon}{_B_C}\tensor{b}{^{||C}_{||A}}\tensor{c}{^{||B}} - \frac12\tensor{b}{_{||AB}}\tensor{(D^2b)}{^{||B}} - \frac34\tensor{b}{_{||AB}}\tensor{b}{^{||B}} \\
    &\phantom{=}\ + \frac14\tensor{(D^2b)}{_{||A}}D^2b + \frac14\tensor{b}{_{||A}}D^2b - \frac14\tensor{b}{^{||BC}}\left(\tensor{b}{_{||ABC}} - \tensor{R}{^E_B_A_C}\tensor{b}{_{||E}}\right) + \frac14\tensor{b}{^{||BC}}\tensor{b}{_{||ABC}} \\
    &= - \frac12\tensor{c}{_{||AB}}\tensor{(D^2c)}{^{||B}} - \tensor{c}{_{||AB}}\tensor{c}{^{||B}} + \frac14\tensor{(D^2c)}{_{||A}}D^2c + \frac12\tensor{c}{_{||A}}D^2c - \frac12\tensor{\epsilon}{^B^C}\tensor{c}{_{||AB}}\tensor{(D^2b)}{_{||C}} \\
    &\phantom{=}\ - \frac12\tensor{\epsilon}{^B^C}\tensor{c}{_{||AB}}\tensor{b}{_{||C}} + \frac12\tensor{\epsilon}{^C^B}\tensor{c}{_{||AB}}\tensor{b}{_{||C}} + \frac14\tensor{\epsilon}{_A_B}\tensor{(D^2b)}{^{||B}}D^2c + \frac14\tensor{\epsilon}{_A_B}\tensor{b}{^{||B}}D^2c \\
    &\phantom{=}\ - \frac14\tensor{\epsilon}{_B_A}\tensor{b}{^{||B}}D^2c - \frac14\tensor{\epsilon}{_A_C}\tensor{b}{^{||C}_{||B}}\tensor{(D^2c)}{^{||B}} - \frac12\tensor{\epsilon}{_A_C}\tensor{b}{^{||C}_{||B}}\tensor{c}{^{||B}} - \frac14\tensor{\epsilon}{_B_C}\tensor{b}{^{||C}_{||A}}\tensor{(D^2c)}{^{||B}} \\
    &\phantom{=}\ - \frac12\tensor{\epsilon}{_B_C}\tensor{b}{^{||C}_{||A}}\tensor{c}{^{||B}} - \frac12\tensor{b}{_{||AB}}\tensor{(D^2b)}{^{||B}} - \frac34\tensor{b}{_{||AB}}\tensor{b}{^{||B}} + \frac14\tensor{(D^2b)}{_{||A}}D^2b + \frac14\tensor{b}{_{||A}}D^2b \\
    &\phantom{=}\ + \frac14\tensor{b}{_{||A}}D^2b - \frac14\tensor{b}{^{||B}}\tensor{b}{_{||AB}} \\
    &= - \frac12\tensor{c}{_{||AB}}\tensor{(D^2c)}{^{||B}} - \tensor{c}{_{||AB}}\tensor{c}{^{||B}} + \frac14\tensor{(D^2c)}{_{||A}}D^2c + \frac12\tensor{c}{_{||A}}D^2c - \frac12\tensor{\epsilon}{^B^C}\tensor{c}{_{||AB}}\tensor{(D^2b)}{_{||C}} \\
    &\phantom{=}\ - \tensor{\epsilon}{^B^C}\tensor{c}{_{||AB}}\tensor{b}{_{||C}} + \frac14\tensor{\epsilon}{_A_B}\tensor{(D^2b)}{^{||B}}D^2c + \frac12\tensor{\epsilon}{_A_B}\tensor{b}{^{||B}}D^2c - \frac14\tensor{\epsilon}{_A_C}\tensor{b}{^{||BC}}\tensor{(D^2c)}{_{||B}} \\
    &\phantom{=}\ - \frac12\tensor{\epsilon}{_A_C}\tensor{b}{^{||BC}}\tensor{c}{_{||B}} - \frac14\tensor{\epsilon}{^B^C}\tensor{b}{_{||AC}}\tensor{(D^2c)}{_{||B}} - \frac12\tensor{\epsilon}{^B^C}\tensor{b}{_{||AC}}\tensor{c}{_{||B}} - \frac12\tensor{b}{_{||AB}}\tensor{(D^2b)}{^{||B}} \\
    &\phantom{=}\ - \tensor{b}{_{||AB}}\tensor{b}{^{||B}} + \frac14\tensor{(D^2b)}{_{||A}}D^2b + \frac12\tensor{b}{_{||A}}D^2b
\end{split}
\end{equation*}

\newpage
Now, we integrate the above expression with $\tensor{\epsilon}{^A^D}\tensor{n}{_{||D}}$, using integration by parts where needed.

\begin{equation*}
\begin{split}
    &- \frac14\int_{S^2} \tensor{\chi}{_A_B}\tensor{\chi}{^B^C_{||C}}\tensor{\epsilon}{^A^D}\tensor{n}{_{||D}} \\
    &= \int_{S^2} \left( - \frac12\tensor{\epsilon}{^A^D}\tensor{n}{_{||D}}\tensor{c}{_{||AB}}\tensor{(D^2c)}{^{||B}} - \tensor{\epsilon}{^A^D}\tensor{n}{_{||D}}\tensor{c}{_{||AB}}\tensor{c}{^{||B}} + \frac14\tensor{\epsilon}{^A^D}\tensor{n}{_{||D}}\tensor{(D^2c)}{_{||A}}D^2c \right. \\
    &\phantom{=}\ \left. + \frac12\tensor{\epsilon}{^A^D}\tensor{n}{_{||D}}\tensor{c}{_{||A}}D^2c - \frac12\tensor{\epsilon}{^A^D}\tensor{n}{_{||D}}\tensor{\epsilon}{^B^C}\tensor{c}{_{||AB}}\tensor{(D^2b)}{_{||C}} - \tensor{\epsilon}{^A^D}\tensor{n}{_{||D}}\tensor{\epsilon}{^B^C}\tensor{c}{_{||AB}}\tensor{b}{_{||C}} \right. \\
    &\phantom{=}\ \left. + \frac14\tensor{\epsilon}{^A^D}\tensor{n}{_{||D}}\tensor{\epsilon}{_A_B}\tensor{(D^2b)}{^{||B}}D^2c + \frac12\tensor{\epsilon}{^A^D}\tensor{n}{_{||D}}\tensor{\epsilon}{_A_B}\tensor{b}{^{||B}}D^2c - \frac14\tensor{\epsilon}{^A^D}\tensor{n}{_{||D}}\tensor{\epsilon}{_A_C}\tensor{b}{^{||BC}}\tensor{(D^2c)}{_{||B}} \right. \\
    &\phantom{=}\ \left. - \frac12\tensor{\epsilon}{^A^D}\tensor{n}{_{||D}}\tensor{\epsilon}{_A_C}\tensor{b}{^{||BC}}\tensor{c}{_{||B}} - \frac14\tensor{\epsilon}{^A^D}\tensor{n}{_{||D}}\tensor{\epsilon}{^B^C}\tensor{b}{_{||AC}}\tensor{(D^2c)}{_{||B}} - \frac12\tensor{\epsilon}{^A^D}\tensor{n}{_{||D}}\tensor{\epsilon}{^B^C}\tensor{b}{_{||AC}}\tensor{c}{_{||B}} \right. \\
    &\phantom{=}\ \left. - \frac12\tensor{\epsilon}{^A^D}\tensor{n}{_{||D}}\tensor{b}{_{||AB}}\tensor{(D^2b)}{^{||B}} - \tensor{\epsilon}{^A^D}\tensor{n}{_{||D}}\tensor{b}{_{||AB}}\tensor{b}{^{||B}} + \frac14\tensor{\epsilon}{^A^D}\tensor{n}{_{||D}}\tensor{(D^2b)}{_{||A}}D^2b \right. \\
    &\phantom{=}\ \left. + \frac12\tensor{\epsilon}{^A^D}\tensor{n}{_{||D}}\tensor{b}{_{||A}}D^2b \right) \\
    &= \int_{S^2} \left( \frac12\tensor{\epsilon}{^A^D}\tensor{n}{_{||D}}\tensor{(D^2c)}{_{||A}}D^2c + \frac12\tensor{\epsilon}{^A^D}\tensor{n}{_{||D}}\tensor{c}{_{||A}}D^2c - \frac12\tensor{\epsilon}{^A^D}\tensor{n}{_{||D}}(\tensor{c}{_{||B}}\tensor{c}{^{||B}})_{||A} \right. \\
    &\phantom{=}\ \left. + \frac18\tensor{\epsilon}{^A^D}\tensor{n}{_{||D}}\tensor{((D^2c)^2)}{_{||A}} + \frac12\tensor{\epsilon}{^A^D}\tensor{n}{_{||D}}\tensor{c}{_{||A}}D^2c - \frac12\left(\tensor{\delta}{_B^A}\tensor{\delta}{_C^D} - \tensor{\delta}{_B^D}\tensor{\delta}{_C^A}\right)\tensor{n}{_{||D}}\tensor{c}{_{||A}^{||B}}\tensor{(D^2b)}{^{||C}} \right. \\
    &\phantom{=}\ \left. - \left(\tensor{\delta}{_B^A}\tensor{\delta}{_C^D} - \tensor{\delta}{_B^D}\tensor{\delta}{_C^A}\right)\tensor{n}{_{||D}}\tensor{c}{_{||A}^{||B}}\tensor{b}{^{||C}} + \frac14\tensor{n}{_{||B}}\tensor{(D^2b)}{^{||B}}D^2c + \frac12\tensor{n}{_{||B}}\tensor{b}{^{||B}}D^2c \right. \\
    &\phantom{=}\ \left. - \frac14\tensor{n}{_{||C}}\tensor{b}{^{||BC}}\tensor{(D^2c)}{_{||B}} - \frac12\tensor{n}{_{||C}}\tensor{b}{^{||BC}}\tensor{c}{_{||B}} - \frac14\left(\tensor{\delta}{_B^A}\tensor{\delta}{_C^D} - \tensor{\delta}{_B^D}\tensor{\delta}{_C^A}\right)\tensor{n}{_{||D}}\tensor{b}{_{||A}^{||C}}\tensor{(D^2c)}{^{||B}} \right. \\
    &\phantom{=}\ \left. - \frac12\left(\tensor{\delta}{_B^A}\tensor{\delta}{_C^D} - \tensor{\delta}{_B^D}\tensor{\delta}{_C^A}\right)\tensor{n}{_{||D}}\tensor{b}{_{||A}^{||C}}\tensor{c}{^{||B}} + \frac12\tensor{\epsilon}{^A^D}\tensor{n}{_{||D}}\tensor{(D^2b)}{_{||A}}D^2b + \frac12\tensor{\epsilon}{^A^D}\tensor{n}{_{||D}}\tensor{b}{_{||A}}D^2b \right. \\
    &\phantom{=}\ \left. - \frac12\tensor{\epsilon}{^A^D}\tensor{n}{_{||D}}(\tensor{b}{_{||B}}\tensor{b}{^{||B}})_{||A} + \frac18\tensor{\epsilon}{^A^D}\tensor{n}{_{||D}}\tensor{((D^2b)^2)}{_{||A}} + \frac12\tensor{\epsilon}{^A^D}\tensor{n}{_{||D}}\tensor{b}{_{||A}}D^2b \right) \\
    &= \int_{S^2} \left( \tensor{\epsilon}{^A^D}\tensor{n}{_{||D}}\tensor{c}{_{||A}}D^2c - \frac12\tensor{n}{_{||C}}\tensor{(D^2b)}{^{||C}}D^2c + \frac12\tensor{n}{^{||B}}\tensor{c}{_{||AB}}\tensor{(D^2b)}{^{||A}} - \tensor{n}{_{||C}}\tensor{b}{^{||C}}D^2c \right. \\
    &\phantom{=}\ \left. + \tensor{n}{_{||B}}\tensor{c}{_{||A}^{||B}}\tensor{b}{^{||A}} + \frac14\tensor{n}{_{||A}}\tensor{(D^2b)}{^{||A}}D^2c + \frac12\tensor{n}{_{||A}}\tensor{b}{^{||A}}D^2c - \frac14\tensor{n}{_{||B}}\tensor{b}{^{||AB}}\tensor{(D^2c)}{_{||A}} \right. \\
    &\phantom{=}\ \left. - \frac12\tensor{n}{_{||B}}\tensor{b}{^{||AB}}\tensor{c}{_{||A}} - \frac14\tensor{n}{_{||C}}\tensor{b}{_{||A}^{||C}}\tensor{(D^2c)}{^{||A}} + \frac14\tensor{n}{_{||B}}\tensor{(D^2c)}{^{||B}}D^2b - \frac12\tensor{n}{_{||C}}\tensor{b}{_{||A}^{||C}}\tensor{c}{^{||A}} \right. \\
    &\phantom{=}\ \left. + \frac12\tensor{n}{_{||B}}\tensor{c}{^{||B}}D^2b + \tensor{\epsilon}{^A^D}\tensor{n}{_{||D}}\tensor{b}{_{||A}}D^2b \right) \\
    &= \int_{S^2} \left( - \frac14\tensor{n}{_{||A}}\tensor{(D^2b)}{^{||A}}D^2c + \frac12\tensor{n}{_{||B}}\tensor{c}{^{||AB}}\tensor{(D^2b)}{_{||A}} - \frac12\tensor{n}{_{||A}}\tensor{b}{^{||A}}D^2c + \tensor{n}{_{||B}}\tensor{c}{^{||AB}}\tensor{b}{_{||A}} \right. \\
    &\phantom{=}\ \left. - \frac12\tensor{n}{_{||B}}\tensor{b}{^{||AB}}\tensor{(D^2c)}{_{||A}} - \tensor{n}{_{||B}}\tensor{b}{^{||AB}}\tensor{c}{_{||A}} + \frac14\tensor{n}{_{||A}}\tensor{(D^2c)}{^{||A}}D^2b + \frac12\tensor{n}{_{||A}}\tensor{c}{^{||A}}D^2b \right)
\end{split}
\end{equation*}
\begin{equation*}
\begin{split}
    &= \int_{S^2} \left( - \frac12nD^2bD^2c + \frac14n_{||A}D^2b(D^2c)^{||A} + \frac12nD^2bD^2c - \frac12n_{||B}D^2b(D^2c)^{||B} - \frac12n_{||B}c^{||B}D^2b \right. \\
    &\phantom{=}\ \left. - \frac12nb_{||A}c^{||A} + \frac12n_{||A}b^{||AB}c_{||B} + nb^{||A}c_{||A} - n_{||B}c^{||B}D^2b - \frac12nD^2bD^2c + \frac12n_{||B}(D^2b)^{||B}D^2c \right. \\
    &\phantom{=}\ \left. + \frac12n_{||B}b^{||B}D^2c - ncD^2b + n_{||B}c(D^2b)^{||B} + n_{||B}cb^{||B} + \frac14n_{||A}D^2b(D^2c)^{||A} + \frac12n_{||A}c^{||A}D^2b \right) \\
    &= \int_{S^2} \left( - \frac12nD^2bD^2c - n_{||A}c^{||A}D^2b + \frac12n_{||A}b^{||A}c + \frac12ncD^2b + \frac12nb^{||A}c_{||A} - \frac12n_{||A}b^{||A}D^2c \right. \\
    &\phantom{=}\ \left. + nb^{||A}c_{||A} + nD^2bD^2c - \frac12n_{||A}D^2b(D^2c)^{||A} + \frac12n_{||A}b^{||A}D^2c - ncD^2b + 2ncD^2b \right. \\
    &\phantom{=}\ \left. - n_{||A}c^{||A}D^2b + n_{||A}b^{||A}c \right) \\
    &= \int_{S^2} \left( - 2n_{||A}c^{||A}D^2b + \frac12nD^2bD^2c - \frac12n_{||A}D^2b(D^2c)^{||A} + \frac32n_{||A}b^{||A}c + \frac32ncD^2b \right. \\
    &\phantom{=}\ \left. \frac32nb^{||A}c_{||A} \right) \\
    &= \int_{S^2} \left( - 2n_{||A}c^{||A}D^2b + \frac12nD^2bD^2c - \frac12n_{||A}D^2b(D^2c)^{||A} + \frac32n_{||A}b^{||A}c + \frac32ncD^2b \right. \\
    &\phantom{=}\ \left. - \frac32n_{||A}b^{||A}c - \frac32ncD^2b \right) \\
    &= \int_{S^2} \left( - 2n_{||A}c^{||A}D^2b + \frac12nD^2bD^2c - \frac12n_{||A}D^2b(D^2c)^{||A} \right)
\end{split}
\end{equation*}

\vspace{1em}
Hence, we get \eqref{eq:cjk-pot}. \eqref{eq:cwy-pot} is obtained by only adding the correction term to \eqref{eq:cjk-pot}.  Finally, we obtain \eqref{eq:cn-pot} by noticing that
\begin{equation*}
\begin{split}
    \cn &= \frac{1}{8\pi}\int_{S^2} \left( - 3\tensor{N}{_A}\tensor{\epsilon}{^A^D}\tensor{n}{_{||D}} - \frac14\tensor{\chi}{_A_B}\tensor{\chi}{^B^C_{||C}}\tensor{\epsilon}{^A^D}\tensor{n}{_{||D}} + \frac{1 - \alpha}{4}\tensor{\chi}{_A_B}\tensor{\chi}{^B^C_{||C}}\tensor{\epsilon}{^A^D}\tensor{n}{_{||D}} \right) \\
    &= \frac{1}{8\pi} \int_{S^2} \left( - 3\tensor{N}{_A}\tensor{\epsilon}{^A^D}\tensor{n}{_{||D}} - \frac{\alpha}{4}\tensor{\chi}{_A_B}\tensor{\chi}{^B^C_{||C}}\tensor{\epsilon}{^A^D}\tensor{n}{_{||D}} \right).
\end{split}
\end{equation*}

\section{Fluxes}

\subsection{In terms of $\chi$}\label{app:flux-chi}

In order to obtain formulas for the fluxes, we differentiate the definitions by $u$, and substitute \eqref{eq:n-flux} into the results expressions. Firstly, for $\dcjk$:
\begin{equation*}
\begin{split}
    \dcjk &= \frac{1}{8\pi} \int_{S^2} \left( -3 \dot{N}_A - \frac14 \dot{\chi}_{AB}\tensor{\chi}{^{BC}_{||C}} - \frac14 \chi_{AB}\tensor{\dot{\chi}}{^{BC}_{||C}} \right) \epsilon^{AD} n_{||D} \\
    &= \frac{1}{8\pi} \int_{S^2} \left( m_{||A}\epsilon^{AD}n_{||D} - \frac14 n_{||A}\Tilde{\lambda}^{||A} + \frac34 \epsilon^{AD}n_{||D}\chi_{AB}\tensor{\dot{\chi}}{^B^C_{||C}} + \frac14 \epsilon^{AD}n_{||D}\dot{\chi}^{BC}\chi_{AB||C} \right. \\
    &\phantom{=}\, \left. - \frac14 \epsilon^{AD} n_{||D}\dot{\chi}_{AB}\tensor{\chi}{^{BC}_{||C}} - \frac14 \epsilon^{AD}n_{||D}\chi_{AB}\tensor{\dot{\chi}}{^{BC}_{||C}} \right)
\end{split}
\end{equation*}
\begin{equation*}
\begin{split}
    &= \frac{1}{8\pi} \int_{S^2} \left( m_{||A}\epsilon^{AD}n_{||D} - \frac14 n_{||A}\Tilde{\lambda}^{||A} + \frac12 \epsilon^{AD}n_{||D}\chi_{AB}\tensor{\dot{\chi}}{^B^C_{||C}} + \frac14 \epsilon^{AD}n_{||D}\dot{\chi}^{BC}\chi_{AB||C} \right. \\
    &\phantom{=}\, \left. - \frac14 \epsilon^{AD}n_{||D}\dot{\chi}_{AB}\tensor{\chi}{^{BC}_{||C}} \right) \\
    &= \frac{1}{8\pi} \int_{S^2} \left( \left( m\epsilon^{AD}n_{||D} \right)_{||A} - m\epsilon^{AD}n_{||DA} - \frac14 n_{||A}\epsilon^{EC}\tensor{\chi}{^B_E_{||BC}^{||A}} \right. \\
    &\phantom{=}\, \left. + \left( \frac12 \epsilon^{AD}n_{||D}\chi_{AB}\dot{\chi}^{BC} \right)_{||C} - \frac12 \epsilon^{AD}n_{||DC}\chi_{AB}\dot{\chi}^{BC} - \frac12 \epsilon^{AD}n_{||D}\chi_{AB||C}\dot{\chi}^{BC} \right. \\
    &\phantom{=}\, \left. + \frac14 \epsilon^{AD}n_{||D}\dot{\chi}^{BC}\chi_{AB||C} - \frac14 \epsilon^{AD}n_{||D}\dot{\chi}_{AB}\tensor{\chi}{^{BC}_{||C}} \right) \\
    &= \frac{1}{8\pi} \int_{S^2} \left( \left( - \frac14 n_{||A}\epsilon^{EC}\tensor{\chi}{^B_E_{||BC}} \right)^{||A} + \frac14 \epsilon^{EC}\tensor{\chi}{^B_E_{||BC}}D^2n_i + \frac12 \epsilon^{AD}nh_{CD}\chi_{AB}\dot{\chi}^{BC} \right. \\
    &\phantom{=}\, \left. - \frac14 \epsilon^{AD}n_{||D}\dot{\chi}^{BC}\chi_{AB||C} - \frac14 \epsilon^{AD}n_{||D}\dot{\chi}_{AB}\tensor{\chi}{^{BC}_{||C}} \right) \\
    &= \frac{1}{8\pi} \int_{S^2} \left( - \frac12 n_i\epsilon^{EC}\tensor{\chi}{^B_E_{||BC}} + \frac12 \epsilon^{AD}nh_{CD}\chi_{AB}\dot{\chi}^{BC} - \frac14 \epsilon^{AD}n_{||D}\dot{\chi}^{BC}\chi_{AB||C} \right. \\
    &\phantom{=}\, \left. - \frac14 \epsilon^{AD}n_{||D}h_{AE}h_{BF}\dot{\chi}^{EF}\tensor{\chi}{^{BC}_{||C}} \right) \\
    &= \frac{1}{8\pi} \int_{S^2} \left( - \frac12 n_{||CB}\epsilon^{EC}\tensor{\chi}{^B_E} + \frac12 \epsilon^{AD}nh_{CD}\chi_{AB}\dot{\chi}^{BC} - \frac14 \epsilon^{AD}n_{||D}\dot{\chi}^{BC}\chi_{AB||C} \right. \\
    &\phantom{=}\, \left. - \frac14 \epsilon^{AD}n_{||D}h_{AB}h_{EC}\dot{\chi}^{BC}\tensor{\chi}{^{EF}_{||F}} \right) \\
    &= \frac{1}{8\pi} \int_{S^2} \left( \frac12 nh_{BC}\epsilon^{EC}\tensor{\chi}{^B_E} + \frac12 \epsilon^{AD}nh_{CD}\chi_{AB}\dot{\chi}^{BC} - \frac14 \epsilon^{AD}n_{||D}\dot{\chi}^{BC}\chi_{AB||C} \right. \\
    &\phantom{=}\, \left. - \frac14 \epsilon^{AD}n_{||D}h_{AB}\dot{\chi}^{BC}\tensor{\chi}{_{CE}^{||E}} \right) \\
    &= \frac{1}{8\pi} \int_{S^2} \epsilon^{AD}\dot{\chi}^{BC} \left( \frac12 nh_{CD}\chi_{AB} - \frac14 n_{||D}\chi_{AB||C} - \frac14 n_{||D}h_{AB}\tensor{\chi}{_{CE}^{||E}} \right)
\end{split}
\end{equation*}

Similarly, we treat the additional term of $\cn$.

\begin{equation*}
\begin{split}
    &\phantom{=} \int_{S^2} \epsilon^{AD}n_{||D} \left( - \frac14 \dot{\chi}_{AB}\tensor{\chi}{^B^C_{||C}} - \frac14 \chi_{AB}\tensor{\dot{\chi}}{^B^C_{||C}} \right) \\
    &= \int_{S^2} \left( - \frac14 \epsilon^{AD}n_{||D}h_{AB}\dot{\chi}^{BC}\tensor{\chi}{_{CE}^{||E}} - \frac14 \epsilon^{AD}nh_{CD}\chi_{AB}\dot{\chi}^{BC} + \frac14 \epsilon^{AD}n_{||D}\chi_{AB||C}\dot{\chi}^{BC} \right) \\
    &= \int_{S^2} \epsilon^{AD}\dot{\chi}^{BC} \left( - \frac14 n_{||D}h_{AB}\tensor{\chi}{_{CE}^{||E}} - \frac14 nh_{CD}\chi_{AB} + \frac14 n_{||D}\chi_{AB||C} \right)
\end{split}
\end{equation*}

\newpage
\subsection{In terms of potentials}\label{app:flux-pot}

Since we assume $b = 0$ in Theorem \ref{thm:fluxes}, it is most convenient to differentiate \cref{eq:cjk-pot,eq:cn-pot,eq:cwy-pot}. As a result, it is sufficient to evaluate $\tensor{\dot{N}}{_A}$ in terms of $c$, and integrate it on $S^2$.

\begin{equation*}
\begin{split}
    &\int_{S^2} - 3\dot{N}_A\epsilon^{AD}n_{||D} \\
    &= \int_{S^2} \left( m_{||A}\epsilon^{AD}n_{||D} - \frac14 n_{||A}\Tilde{\lambda}^{||A} + \frac34 \epsilon^{AD}n_{||D}\chi_{AB}\tensor{\dot{\chi}}{^B^C_{||C}} + \frac14 \epsilon^{AD}n_{||D}\dot{\chi}^{BC}\chi_{AB||C} \right) \\
    &= \int_{S^2} \left( \frac34 \epsilon^{AD}n_{||D}(2c_{||AB} - h_{AB}D^2c)(2\dot{c}^{||BC} - h^{BC}D^2\dot{c})_{||C} \right. \\
    &\phantom{=}\ \left. + \frac14 \epsilon^{AD}n_{||D}(2\dot{c}^{||BC} - h^{BC}D^2\dot{c})(2c_{||AB} - h_{AB}D^2c)_{||C} \right) \\
    &= \int_{S^2} \left( 3\epsilon^{AD}n_{||D}c_{||AB}\tensor{\dot{c}}{^{||BC}_{||C}} - \frac32\epsilon^{AD}n_{||D}c_{||AB}(D^2\dot{c})^{||B} - \frac32\epsilon^{AD}n_{||D}\tensor{\dot{c}}{_{||A}^{||C}_{||C}}D^2c \right. \\
    &\phantom{=}\ \left. + \frac34\epsilon^{AD}n_{||D}(D^2\dot{c})_{||A}D^2c + \epsilon^{AD}n_{||D}\dot{c}^{||BC}c_{||ABC} - \frac12\epsilon^{AD}n_{||D}\tensor{\dot{c}}{_{||A}^{||C}}(D^2c)_{||C} \right. \\
    &\phantom{=}\ \left. - \frac12\epsilon^{AD}n_{||D}\tensor{c}{_{||A}^{||C}_{||C}}D^2\dot{c} + \frac14\epsilon^{AD}n_{||D}(D^2c)_{||A}D^2\dot{c} \right) \\
    &= \int_{S^2} \left( 3\epsilon^{AD}n_{||D}c_{||AB}(D^2\dot{c})^{||B} + 3\epsilon^{AD}n_{||D}c_{||AB}\dot{c}^{||B} - \frac32\epsilon^{AD}n_{||D}c_{||AB}(D^2\dot{c})^{||B} \right. \\
    &\phantom{=}\ \left. - \frac32\epsilon^{AD}n_{||D}(D^2\dot{c})_{||A}D^2c - \frac32\epsilon^{AD}n_{||D}\dot{c}_{||A}D^2c + \frac34\epsilon^{AD}n_{||D}(D^2\dot{c})_{||A}D^2c \right. \\
    &\phantom{=}\ \left. + \epsilon^{AD}n_{||D}\dot{c}^{||BC}c_{||ABC} - \frac12\epsilon^{AD}n_{||D}\tensor{\dot{c}}{_{||A}^{||C}}(D^2c)_{||C} - \frac12\epsilon^{AD}n_{||D}(D^2c)_{||A}D^2\dot{c} \right. \\
    &\phantom{=}\ \left. - \frac12\epsilon^{AD}n_{||D}c_{||A}D^2\dot{c} + \frac14\epsilon^{AD}n_{||D}(D^2c)_{||A}D^2\dot{c} \right) \\
    &= \int_{S^2} \left( \frac32\epsilon^{AD}n_{||D}c_{||AB}(D^2\dot{c})^{||B} + 3\epsilon^{AD}n_{||D}c_{||AB}\dot{c}^{||B} - \frac34\epsilon^{AD}n_{||D}(D^2\dot{c})_{||A}D^2c \right. \\
    &\phantom{=}\ \left. - \frac32\epsilon^{AD}n_{||D}\dot{c}_{||A}D^2c + \epsilon^{AD}n_{||D}\dot{c}^{||BC}c_{||ABC} - \frac12\epsilon^{AD}n_{||D}\tensor{\dot{c}}{_{||AB}}(D^2c)^{||B} \right. \\
    &\phantom{=}\ \left. - \frac14\epsilon^{AD}n_{||D}(D^2c)_{||A}D^2\dot{c} - \frac12\epsilon^{AD}n_{||D}c_{||A}D^2\dot{c} \right) \\
    &= \int_{S^2} \left( \frac32\epsilon^{AD}n_{||D}c_{||AB}(D^2\dot{c})^{||B} + 3\epsilon^{AD}n_{||D}c_{||AB}\dot{c}^{||B} - \frac34\epsilon^{AD}n_{||D}(D^2\dot{c})_{||A}D^2c \right. \\
    &\phantom{=}\ \left. - \frac32\epsilon^{AD}n_{||D}\dot{c}_{||A}D^2c + \left(\epsilon^{AD}n_{||D}\dot{c}^{||BC}c_{||AB}\right)_{||C} - \epsilon^{AD}n_{||DC}\dot{c}^{||BC}c_{||AB} \right. \\
    &\phantom{=}\ \left. - \epsilon^{AD}n_{||D}\tensor{\dot{c}}{^{||BC}_{||C}}c_{||AB} - \frac12\epsilon^{AD}n_{||D}\tensor{\dot{c}}{_{||AB}}(D^2c)^{||B} - \frac14\epsilon^{AD}n_{||D}(D^2c)_{||A}D^2\dot{c} \right. \\
    &\phantom{=}\ \left. - \frac12\epsilon^{AD}n_{||D}c_{||A}D^2\dot{c} \right)
\end{split}
\end{equation*}
\begin{equation*}
\begin{split}
    &= \int_{S^2} \left( \frac32\epsilon^{AD}n_{||D}c_{||AB}(D^2\dot{c})^{||B} + 3\epsilon^{AD}n_{||D}c_{||AB}\dot{c}^{||B} - \frac34\epsilon^{AD}n_{||D}(D^2\dot{c})_{||A}D^2c \right. \\
    &\phantom{=}\ \left. - \frac32\epsilon^{AD}n_{||D}\dot{c}_{||A}D^2c + \epsilon^{AD}nh_{CD}\dot{c}^{||BC}c_{||AB} - \epsilon^{AD}n_{||D}c_{||AB}(D^2\dot{c})^{||B} \right. \\
    &\phantom{=}\ \left. - \epsilon^{AD}n_{||D}c_{||AB}\dot{c}^{||B} - \frac12\epsilon^{AD}n_{||D}\tensor{\dot{c}}{_{||AB}}(D^2c)^{||B} - \frac14\epsilon^{AD}n_{||D}(D^2c)_{||A}D^2\dot{c} \right. \\
    &\phantom{=}\ \left. - \frac12\epsilon^{AD}n_{||D}c_{||A}D^2\dot{c} \right) \\
    &= \int_{S^2} \left( \frac12\epsilon^{AD}n_{||D}c_{||AB}(D^2\dot{c})^{||B} + 2\epsilon^{AD}n_{||D}c_{||AB}\dot{c}^{||B} - \frac34\epsilon^{AD}n_{||D}(D^2\dot{c})_{||A}D^2c \right. \\
    &\phantom{=}\ \left. - \frac32\epsilon^{AD}n_{||D}\dot{c}_{||A}D^2c + \epsilon^{AD}nc_{||AB}\tensor{\dot{c}}{^{||B}_{||D}} - \frac12\epsilon^{AD}n_{||D}\tensor{\dot{c}}{_{||AB}}(D^2c)^{||B} \right. \\
    &\phantom{=}\ \left. - \frac14\epsilon^{AD}n_{||D}(D^2c)_{||A}D^2\dot{c} - \frac12\epsilon^{AD}n_{||D}c_{||A}D^2\dot{c} \right) \\
    &= \int_{S^2} \left( \left( \frac12\epsilon^{AD}n_{||D}\tensor{c}{_{||A}^{||B}}(D^2\dot{c}) \right)_{||B} - \frac12\epsilon^{AD}n_{||D}(D^2c)_{||A}D^2\dot{c} - \frac12\epsilon^{AD}n_{||D}c_{||A}D^2\dot{c} \right. \\
    &\phantom{=}\ \left. + 2\epsilon^{AD}n_{||D}c_{||AB}\dot{c}^{||B} - \frac34\epsilon^{AD}n_{||D}(D^2\dot{c})_{||A}D^2c - \frac32\epsilon^{AD}n_{||D}\dot{c}_{||A}D^2c + \left( \epsilon^{AD}nc_{||AB}\dot{c}^{||B} \right)_{||D} \right. \\
    &\phantom{=}\ \left. - \epsilon^{AD}n_{||D}c_{||AB}\dot{c}^{||B} - \epsilon^{AD}nc_{||ABD} - \left( \frac12\epsilon^{AD}n_{||D}\tensor{\dot{c}}{_{||A}^{||B}}(D^2c) \right)_{||B} + \frac12\epsilon^{AD}n_{||D}(D^2\dot{c})_{||A}D^2c \right. \\
    &\phantom{=}\ \left. + \frac12\epsilon^{AD}n_{||D}\dot{c}_{||A}D^2c - \frac14\epsilon^{AD}n_{||D}(D^2c)_{||A}D^2\dot{c} - \frac12\epsilon^{AD}n_{||D}c_{||A}D^2\dot{c} \right) \\
    &= \int_{S^2} \left( - \frac34\epsilon^{AD}n_{||D}(D^2c)_{||A}D^2\dot{c} - \epsilon^{AD}n_{||D}c_{||A}D^2\dot{c} + \epsilon^{AD}n_{||D}c_{||AB}\dot{c}^{||B} \right. \\
    &\phantom{=}\ \left. - \frac14\epsilon^{AD}n_{||D}(D^2\dot{c})_{||A}D^2c - \epsilon^{AD}n_{||D}\dot{c}_{||A}D^2c - \epsilon^{AD}n\dot{c}^{||B}\left(c_{||ADB} - \tensor{R}{^C_A_D_B}c_{||C}\right) \right) \\
    &= \int_{S^2} \left( - \frac34\epsilon^{AD}n_{||D}(D^2c)_{||A}D^2\dot{c} - \epsilon^{AD}n_{||D}c_{||A}D^2\dot{c} + \left(\epsilon^{AD}n_{||D}c_{||A}\dot{c}^{||B}\right)_{||B} \right. \\
    &\phantom{=}\ \left. + \epsilon^{AD}nc_{||A}\dot{c}_{||D} - \epsilon^{AD}n_{||D}c_{||A}D^2\dot{c} - \frac14\epsilon^{AD}n_{||D}(D^2\dot{c})_{||A}D^2c - \epsilon^{AD}n_{||D}\dot{c}_{||A}D^2c \right. \\
    &\phantom{=}\ \left. + \epsilon^{AD}n\dot{c}^{||B}c_{||C}\left(\tensor{\delta}{^C_D}h_{AB} - \tensor{\delta}{^C_B}h_{AD}\right) \vphantom{\frac12}\right) \\
    &= \int_{S^2} \left( - \frac34\epsilon^{AD}n_{||D}(D^2c)_{||A}D^2\dot{c} - 2\epsilon^{AD}n_{||D}c_{||A}D^2\dot{c} + \epsilon^{AD}nc_{||A}\dot{c}_{||D} - \frac14\epsilon^{AD}n_{||D}(D^2\dot{c})_{||A}D^2c \right. \\
    &\phantom{=}\ \left. - \epsilon^{AD}n_{||D}\dot{c}_{||A}D^2c + \epsilon^{AD}nc_{||D}\dot{c}_{||A} \vphantom{\frac12}\right) \\
    &= \int_{S^2} \left( - \frac34\epsilon^{AD}n_{||D}(D^2c)_{||A}D^2\dot{c} - 2\epsilon^{AD}n_{||D}c_{||A}D^2\dot{c} + \epsilon^{AD}nc_{||A}\dot{c}_{||D} - \frac14\epsilon^{AD}n_{||D}(D^2\dot{c})_{||A}D^2c \right. \\
    &\phantom{=}\ \left. - \epsilon^{AD}n_{||D}\dot{c}_{||A}D^2c - \epsilon^{AD}nc_{||A}\dot{c}_{||D} \vphantom{\frac12}\right)
\end{split}
\end{equation*}
\begin{equation*}
\begin{split}
    &= \int_{S^2} \left( - \frac34\epsilon^{AD}n_{||D}(D^2c)_{||A}D^2\dot{c} - 2\epsilon^{AD}n_{||D}c_{||A}D^2\dot{c} - \left(\frac14\epsilon^{AD}n_{||D}D^2cD^2\dot{c}\right)_{||A} \right. \\
    &\phantom{=}\ \left. + \frac14\epsilon^{AD}n_{||D}(D^2c)_{||A}D^2\dot{c} - \left(\epsilon^{AD}n_{||D}c^{||B}\dot{c}_{||A}\right)_{||B} - \epsilon^{AD}nc_{||D}\dot{c}_{||A} + \epsilon^{AD}n_{||D}c^{||B}\dot{c}_{||AB} \right) \\
    &= \int_{S^2} \left( - \frac12\epsilon^{AD}n_{||D}(D^2c)_{||A}D^2\dot{c} - 2\epsilon^{AD}n_{||D}c_{||A}D^2\dot{c} - \epsilon^{AD}nc_{||D}\dot{c}_{||A} + \left(\epsilon^{AD}n_{||D}c\tensor{\dot{c}}{_{||A}^{||B}}\right)_{||B} \right. \\
    &\phantom{=}\ \left. - \epsilon^{AD}n_{||D}c(D^2\dot{c})_{||A} - \epsilon^{AD}n_{||D}c\dot{c}_{||A} \vphantom{\frac12}\right) \\
    &= \int_{S^2} \left( - \frac12\epsilon^{AD}n_{||D}(D^2c)_{||A}D^2\dot{c} - 2\epsilon^{AD}n_{||D}c_{||A}D^2\dot{c} - \epsilon^{AD}nc_{||D}\dot{c}_{||A} - \left(\epsilon^{AD}n_{||D}cD^2\dot{c}\right)_{||A} \right. \\
    &\phantom{=}\ \left. + \epsilon^{AD}n_{||D}c_{||A}D^2\dot{c} - \left(\epsilon^{AD}nc\dot{c}_{||A}\right)_{||D} + \epsilon^{AD}nc_{||D}\dot{c}_{||A} \vphantom{\frac12}\right) \\
    &= \int_{S^2} \left( - \frac12\epsilon^{AD}n_{||D}(D^2c)_{||A}D^2\dot{c} - \epsilon^{AD}n_{||D}c_{||A}D^2\dot{c} \right) \\
    &= \int_{S^2} \left( - \frac12\epsilon^{AD}n_{||D}\tensor{c}{_{||A}^{||B}_{||B}}D^2\dot{c} + \frac12\epsilon^{AD}n_{||D}c_{||A}D^2\dot{c} - \epsilon^{AD}n_{||D}c_{||A}D^2\dot{c} \right) \\
    &= \int_{S^2} \left( - \left(\frac12\epsilon^{AD}n_{||D}\tensor{c}{_{||A}^{||B}}D^2\dot{c}\right)_{||B} + \frac12\epsilon^{AD}n_{||D}\tensor{c}{_{||A}^{||B}}(D^2\dot{c})_{||B} - \frac12\epsilon^{AD}n_{||D}c_{||A}D^2\dot{c} \right) \\
    &= \int_{S^2} \left( \left(\frac12\epsilon^{AD}n_{||D}c_{||A}(D^2\dot{c})^{||B}\right)_{||B} + \frac12\epsilon^{AD}nc_{||A}(D^2\dot{c})_{||D} - \frac12\epsilon^{AD}n_{||D}c_{||A}D^2D^2\dot{c} \right. \\
    &\phantom{=}\ \left. - \frac12\epsilon^{AD}n_{||D}c_{||A}D^2\dot{c} \right) \\
    &= \int_{S^2} \left( \left(\frac12\epsilon^{AD}nc_{||A}D^2\dot{c}\right)_{||D} - \frac12\epsilon^{AD}n_{||D}c_{||A}D^2\dot{c} - \frac12\epsilon^{AD}n_{||D}c_{||A}D^2D^2\dot{c} \right. \\
    &\phantom{=}\ \left. - \frac12\epsilon^{AD}n_{||D}c_{||A}D^2\dot{c} \right) \\
    &= \int_{S^2} \left( - \frac12\epsilon^{AD}n_{||D}c_{||A}D^2D^2\dot{c} - \epsilon^{AD}n_{||D}c_{||A}D^2\dot{c} \right) = - \frac12\int_{S^2} \epsilon^{AD}n_{||D}c_{||A} D^2(D^2 + 2)\dot{c}
\end{split}
\end{equation*}

This is the same for $\dcn$ since $\cjk = \cn$ for $b=0$. For $\dcwy$, we just need to add the derivative of the correction term, which is evaluated as follows.

\begin{equation*}
\begin{split}
    &\phantom{=} \int_{S^2} \left( 2\dot{m}\epsilon^{AD}n_{i||D}c_{||A} + 2m\epsilon^{AD}n_{i||D}\dot{c}_{||A} \right) \\
    &= \int_{S^2} \left( 2\epsilon^{AD}n_{||D}c_{||A} \left( - \frac18 h^{BE}h^{CF}\dot{\chi}_{BC}\dot{\chi}_{EF} + \frac14 \tensor{\dot{\chi}}{^B^C_{||CB}} \right) + 2\epsilon^{AD}n_{||D}m\dot{c}_{||A} \right) \\
    &= \int_{S^2} \left( - \frac14 \epsilon^{AD}n_{||D}c_{||A} \left( 2\dot{c}^{||BC} - h^{BC}D^2\dot{c} \right) \left( 2\dot{c}_{||BC} - h_{BC}D^2\dot{c} \right) \right. \\
    &\phantom{=}\, \left. + \frac12 \epsilon^{AD}n_{||D}c_{||A} \left( 2\dot{c}^{||BC} - h^{BC}D^2\dot{c} \right)_{||CB} + \vphantom{\frac12} 2\epsilon^{AD}n_{||D}m\dot{c}_{||A} \right)
\end{split}
\end{equation*}
\begin{equation*}
\begin{split}
    &= \int_{S^2} \left( - \epsilon^{AD}n_{||D}c_{||A}\dot{c}^{||BC}\dot{c}_{||BC} + \frac12\epsilon^{AD}n_{||D}c_{||A}h_{BC}\dot{c}^{||BC}D^2\dot{c} + \frac12\epsilon^{AD}n_{||D}c_{||A}h^{BC}\dot{c}_{||BC}D^2\dot{c} \right. \\
    &\phantom{=}\, \left. - \frac14\epsilon^{AD}n_{||D}c_{||A}h^{BC}h_{BC}(D^2\dot{c})^2 + \epsilon^{AD}n_{||D}c_{||A}\tensor{\dot{c}}{^{||BC}_{||CB}} - \frac12 \epsilon^{AD}n_{||D}c_{||A}h^{BC}(D^2\dot{c})_{||CB} \right. \\
    &\phantom{=}\, \left. + 2\epsilon^{AD}n_{||D}m\dot{c}_{||A} \vphantom{\frac12}\right) \\
    &= \int_{S^2} \left( - \epsilon^{AD}n_{||D}c_{||A}\dot{c}^{||BC}\dot{c}_{||BC} + \frac12\epsilon^{AD}n_{||D}c_{||A}(D^2\dot{c})^2 + \epsilon^{AD}n_{||D}c_{||A}\tensor{\dot{c}}{^{||BC}_{||CB}} \right. \\
    &\phantom{=}\, \left. - \frac12 \epsilon^{AD}n_{||D}c_{||A}D^2D^2\dot{c} + 2\epsilon^{AD}n_{||D}m\dot{c}_{||A} \right) \\
    &= \int_{S^2} \left( - \epsilon^{AD}n_{||D}c_{||A}\dot{c}^{||BC}\dot{c}_{||BC} + \frac12\epsilon^{AD}n_{||D}c_{||A}(D^2\dot{c})^2 + \frac12 \epsilon^{AD}n_{||D}c_{||A}D^2(D^2 + 2)\dot{c} \right. \\
    &\phantom{=}\, \left. + 2\epsilon^{AD}n_{||D}m\dot{c}_{||A} \vphantom{\frac12}\right)
\end{split}
\end{equation*}

By multiplying the above term respectively by $\frac32$ (or $\frac12$) and adding it to $\dcjk$, one also obtains the flux of the proposed angular momentum $\dot{J}$ (or $\dot{\Tilde{J}}$). 

\section{Transformations under supertranslations}\label{app:transformations}

First, let us notice that, in view of \eqref{eq:chi} and \eqref{eq:chi-trans}, $b \mapsto b$ and $c \mapsto c - f$ under any supertranslation $u \mapsto u - f$. Therefore, to prove Theorem \ref{thm:supertranslations}, it suffices to substitute these into the formulae from Lemma \ref{lem:momenta-potentials} and integrate all terms by parts so that each term contains $f$ (without derivative).

\begin{equation*}
\begin{split}
    \cjk &\longmapsto \frac{1}{8\pi}\int_{S^2} \left( - 3\epsilon^{AD}n_{||D} \left( N_A - mf_{||A} + \frac14f_{||C}\tensor{\chi}{_A_B^{||BC}} - \frac14f_{||C}\tensor{\chi}{^C^B_{||BA}} \right) \right. \\
    &\phantom{=}\quad \left. - 2n_{||A}(c-f)^{||A}D^2b + \frac12nD^2bD^2(c-f) - \frac12n_{||A}D^2b(D^2(c-f))^{||A} \right) \\
    &= \frac{1}{8\pi}\int_{S^2} \left( - 3\epsilon^{AD}n_{||D}N_A + 3\epsilon^{AD}n_{||D}mf_{||A} - \frac34\epsilon^{AD}n_{||D}f_{||C} \left( 2\tensor{c}{_{||AB}^{||BC}} \right. \right. \\
    &\phantom{=}\ \left. \left. - h_{AB}(D^2c)^{||BC} + \epsilon_{AE}\tensor{b}{^{||E}_{||B}^{||BC}} + \epsilon_{BE}\tensor{b}{^{||E}_{||A}^{||BC}} \right) \right. \\
    &\phantom{=}\ \left. + \frac34\epsilon^{AD}n_{||D}f_{||C} \left( 2\tensor{c}{^{||CB}_{||BA}} - h^{CB}(D^2c)_{||BA} + \epsilon^{CE}\tensor{b}{_{||E}^{||B}_{||BA}} + \epsilon^{BE}\tensor{b}{_{||E}^{||C}_{||BA}} \right) \right. \\
    &\phantom{=}\ \left. - 2n_{||A}c^{||A}D^2b + 2n_{||A}f^{||A}D^2b + \frac12nD^2bD^2c - \frac12nD^2bD^2f - \frac12n_{||A}D^2b(D^2c)^{||A} \right. \\
    &\phantom{=}\ \left. + \frac12n_{||A}D^2b(D^2f)^{||A} \right)
\end{split}
\end{equation*}
\begin{equation*}
\begin{split}
    &= \cjk + \frac{1}{8\pi}\int_{S^2} \left( 3\epsilon^{AD}n_{||D}mf_{||A} - \frac32\epsilon^{AD}n_{||D}f_{||C}\tensor{c}{_{||AB}^{||BC}} + \frac34\epsilon^{AD}n_{||D}f_{||C}\tensor{(D^2c)}{_{||A}^{||C}} \right. \\
    &\phantom{=}\ \left. - \frac34\epsilon^{AD}\epsilon_{AE}n_{||D}f_{||C}\tensor{b}{^{||E}_{||B}^{||BC}} - \frac34\epsilon^{AD}\epsilon_{BE}n_{||D}f_{||C}\tensor{b}{^{||E}_{||A}^{||BC}} + \frac32\epsilon^{AD}n_{||D}f_{||C}\tensor{c}{^{||CB}_{||BA}} \right. \\
    &\phantom{=}\ \left. - \frac34\epsilon^{AD}n_{||D}f_{||C}\tensor{(D^2c)}{^{||C}_{||A}} + \frac34\epsilon^{AD}\epsilon_{CE}n_{||D}f^{||C}\tensor{b}{^{||EB}_{||BA}} + \frac34\epsilon^{AD}\epsilon_{BE}n_{||D}f_{||C}\tensor{b}{^{||ECB}_{||A}} \right. \\
    &\phantom{=}\ \left. + 2n_{||A}f^{||A}D^2b - \frac12nD^2bD^2f + \frac12n_{||A}D^2b(D^2f)^{||A} \right) \\
    &= \cjk + \frac{1}{8\pi}\int_{S^2} \left( 3\epsilon^{AD}n_{||D}mf_{||A} - \frac32\epsilon^{AD}n_{||D}f_{||C}\tensor{(D^2c)}{_{||A}^{||C}} - \frac32\epsilon^{AD}n_{||D}f_{||C}\tensor{c}{_{||A}^{||C}} \right. \\
    &\phantom{=}\ \left. - \frac34\tensor{\delta}{_E^D}n_{||D}f_{||C}(D^2b)^{||EC} - \frac34\tensor{\delta}{_E^D}n_{||D}f_{||C}b^{||EC} - \frac34 \left( \tensor{\delta}{_B^A}\tensor{\delta}{_E^D} - \tensor{\delta}{_B^D}\tensor{\delta}{_E^A} \right) n_{||D}f_{||C}\tensor{b}{^{||E}_{||A}^{||BC}} \right. \\
    &\phantom{=}\ \left. + \frac32\epsilon^{AD}n_{||D}f_{||C}\tensor{(D^2c)}{^{||C}_{||A}} + \frac32\epsilon^{AD}n_{||D}f_{||C}\tensor{c}{^{||C}_{||A}} \right. \\
    &\phantom{=}\ \left. + \frac34 \left( \tensor{\delta}{_C^A}\tensor{\delta}{_E^D} - \tensor{\delta}{_C^D}\tensor{\delta}{_E^A} \right) n_{||D}f^{||C} \left( \tensor{(D^2b)}{^{||E}_{||A}} + \tensor{b}{^{||E}_{||A}} \right) \right. \\
    &\phantom{=}\ \left. + \frac34 \left( \tensor{\delta}{_B^A}\tensor{\delta}{_E^D} - \tensor{\delta}{_B^D}\tensor{\delta}{_E^A} \right) n_{||D}f_{||C}\tensor{b}{^{||ECB}_{||A}} + 2n_{||A}f^{||A}D^2b - \frac12nD^2bD^2f \right. \\
    &\phantom{=}\ \left. + \frac12n_{||A}D^2b(D^2f)^{||A} \right) \\
    &= \cjk + \frac{1}{8\pi}\int_{S^2} \left( 3\epsilon^{AD}n_{||D}mf_{||A} - \frac34n_{||A}f_{||B}(D^2b)^{||AB} - \frac34n_{||A}f_{||B}b^{||AB} \right. \\
    &\phantom{=}\ \left. - \frac34n_{||D}f_{||C}\tensor{b}{^{||D}_{||A}^{||AC}} + \frac34n_{||B}f_{||C}\tensor{b}{^{||A}_{||A}^{||BC}} + \frac34n_{||D}f^{||A}\tensor{(D^2b)}{^{||D}_{||A}} + \frac34n_{||D}f^{||A}\tensor{b}{^{||D}_{||A}} \right. \\
    &\phantom{=}\ \left. - \frac34n_{||C}f^{||C}D^2D^2b - \frac34n_{||C}f^{||C}D^2b + \frac34n_{||D}f_{||C}\tensor{b}{^{||DCA}_{||A}} - \frac34n_{||B}f_{||C}\tensor{b}{^{||ACB}_{||A}} \right. \\
    &\phantom{=}\ \left. + 2n_{||A}f^{||A}D^2b - \frac12nD^2bD^2f + \frac12n_{||A}D^2b(D^2f)^{||A} \right) \\
    &= \cjk + \frac{1}{8\pi}\int_{S^2} \left( 3\epsilon^{AD}n_{||D}mf_{||A} - \frac34n_{||A}f_{||B}(D^2b)^{||AB} - \frac34n_{||A}f_{||B}b^{||AB} \right. \\
    &\phantom{=}\ \left. - \frac34n_{||A}f_{||B}(D^2b)^{||AB} - \frac34n_{||A}f_{||B}b^{||AB} + \frac34n_{||A}f_{||B}(D^2b)^{||AB} + \frac34n_{||A}f_{||B}(D^2b)^{||AB} \right. \\
    &\phantom{=}\ \left. + \frac34n_{||A}f_{||B}b^{||AB} - \frac34n_{||A}f^{||A}D^2D^2b - \frac34n_{||A}f^{||A}D^2b + \frac34n_{||B}f_{||C}\tensor{b}{^{||BCA}_{||A}} \right. \\
    &\phantom{=}\ \left. - \frac34n_{||B}f_{||C}\tensor{b}{^{||ACB}_{||A}} + 2n_{||A}f^{||A}D^b - \frac12nD^2bD^2f + \frac12n_{||A}D^2b(D^2f)^{||A} \right) \\
    &= \cjk + \frac{1}{8\pi}\int_{S^2} \left( 3\epsilon^{AD}n_{||D}mf_{||A} - \frac34n_{||A}f_{||B}b^{||AB} - \frac34n_{||A}f^{||A}D^2D^2b \right. \\
    &\phantom{=}\ \left. + \frac54n_{||A}f^{||A}D^2b + \frac34n^{||B}f^{||C} \left( \tensor{b}{_{||BAC}^{||A}} + \tensor{R}{^D_B_C_A}\tensor{b}{_{||D}^{||A}} - \tensor{b}{_{||ABC}^{||A}} - \tensor{R}{^D_A_C_B}\tensor{b}{_{||D}^{||A}} \right) \right. \\
    &\phantom{=}\ \left. - \frac12nD^2bD^2f + \frac12n_{||A}D^2b(D^2f)^{||A} \right)
\end{split}
\end{equation*}
\begin{equation*}
\begin{split}
    &= \cjk + \frac{1}{8\pi} \int_{S^2} \left( 3\epsilon^{AD}n_{||D}mf_{||A} - \frac{3}{4}n_{||A}f_{||B}b^{||AB} - \frac{3}{4}n_{||A}f^{||A}D^2D^2b \right. \\
    &\phantom{=}\ \left. + \frac{5}{4}n_{||A}f^{||A}D^2b + \frac{3}{4}n^{||B}f^{||C}\tensor{b}{_{||D}^{||A}} \left( \tensor{\delta}{^D_C}h_{BA} - \tensor{\delta}{^D_A}h_{BC} - \tensor{\delta}{^D_C}h_{AB} + \tensor{\delta}{^D_B}h_{AC} \right) \right. \\
    &\phantom{=}\ \left. - \frac{1}{2}nD^2bD^2f + \frac{1}{2}n_{||A}D^2b(D^2f)^{||A} \right) \\
    &= \cjk + \frac{1}{8\pi} \int_{S^2} \left( 3\epsilon^{AD}n_{||D}mf_{||A} - \frac{3}{4}n_{||A}f_{||B}b^{||AB} - \frac{3}{4}n_{||A}f^{||A}D^2D^2b \right. \\
    &\phantom{=}\ \left. + \frac{5}{4}n_{||A}f^{||A}D^2b + \frac{3}{4}n^{||B}f_{||A}\tensor{b}{_{||B}^{||A}} - \frac{3}{4}n^{||B}f_{||B}D^2b - \frac{1}{2}nD^2bD^2f + \frac{1}{2}n_{||A}D^2b(D^2f)^{||A} \right) \\
    &= \cjk + \frac{1}{8\pi} \int_{S^2} \left( 3\epsilon^{AD}n_{||D}mf_{||A} - \frac34n_{||A}f^{||A}D^2D^2b + \frac12n_{||A}f^{||A}D^2b - \frac12nD^2bD^2f \right. \\
    &\phantom{=}\ \left. + \frac12n_{||A}D^2b(D^2f)^{||A} \right) \\
    &= \cjk + \frac{1}{8\pi} \int_{S^2} \left( - 3\epsilon^{AD}n_{||D}m_{||A}f - \frac32nfD^2D^2b + \frac34n_{||A}fD^AD^2D^2b + nfD^2b \right. \\
    &\phantom{=}\ \left. - \frac12n_{||A}fD^AD^2b - \frac12nD^2bD^2f + nD^2bD^2f - \frac12n_{||A}D^AD^2bD^2f \right) \\
    &= \cjk + \frac{1}{8\pi} \int_{S^2} \left( - 3\epsilon^{AD}n_{||D}m_{||A}f - \frac32nfD^2D^2b + \frac34n_{||A}fD^AD^2D^2b + nfD^2b \right. \\
    &\phantom{=}\ \left. - \frac12n_{||A}fD^AD^2b - \frac12n_{||A}D^2bf^{||A} - \frac12n(D^2b){||A}f^{||A} - \frac12nf_{||A}D^AD^2b \right. \\
    &\phantom{=}\ \left. + \frac12n_{||A}\tensor{(D^2b)}{^{||A}_{||B}}f^{||B} \right) \\
    &= \cjk + \frac{1}{8\pi} \int_{S^2} \left( - 3\epsilon^{AD}n_{||D}m_{||A}f - \frac32nfD^2D^2b + \frac34n_{||A}fD^AD^2D^2b + nfD^2b \right. \\
    &\phantom{=}\ \left. - \frac12n_{||A}fD^AD^2b - nfD^2b + \frac12n_{||A}fD^AD^2b + \frac12n_{||A}fD^AD^2b + \frac12nfD^2D^2b \right. \\
    &\phantom{=}\ \left. + \frac12n_{||A}fD^AD^2b + \frac12nfD^2D^2b + \frac12nfD^2D^2b - \frac12n_{||A}fD^AD^2D^2b - \frac12n_{||A}fD^AD^2b \right) \\
    &= \cjk + \frac{1}{8\pi} \int_{S^2} f \left( - 3\epsilon^{AD}n_{||D}m_{||A} + \frac14n_{||A}D^AD^2D^2b + \frac12n_{||A}D^AD^2b \right) \\
    &= \cjk + \frac{1}{8\pi} \int_{S^2} f \left( - 3\epsilon^{AD}n_{||D}m_{||A} + \frac14n_{||A}D^AD^2(D^2+2)b \right)
\end{split}
\end{equation*}

A similar procedure is to be performed for $\cn$ and $\cwy$, with additional terms to be taken into consideration. 

\end{document}